\documentclass[11pt,parskip=half]{article}

\usepackage{geometry}                
\usepackage[utf8]{inputenc}
\usepackage[parfill]{parskip}    
\usepackage{graphicx}

\usepackage{amsmath, amssymb, amsthm, amsfonts}
\usepackage{mathtools}
\usepackage{bbm} 

\usepackage[colorlinks,allcolors=blue]{hyperref}

\usepackage[svgnames]{xcolor}
\usepackage[breakable,theorems]{tcolorbox}
\usepackage[font=small,labelfont=bf]{caption}

\usepackage{libertine}
\usepackage[libertine]{newtxmath}
\usepackage{inconsolata} 
\usepackage{microtype}

\usepackage{booktabs}

\usepackage{csquotes}

\usepackage[backend=biber,giveninits=true, maxbibnames=7]{biblatex}
\bibliography{counting}

\makeatletter
\DeclareRobustCommand*\uell{\mathpalette\@uell\relax}
\newcommand*\@uell[2]{
  \setbox0=\hbox{$#1\ell$}
  \setbox1=\hbox{\rotatebox{10}{$#1\ell$}}
  \dimen0=\wd0 \advance\dimen0 by -\wd1 \divide\dimen0 by 2
  \mathord{\lower 0.1ex \hbox{\kern\dimen0\unhbox1\kern\dimen0}}
}

\newcommand{\photons}{Y}
\newcommand{\ccdvals}{\widetilde{Y}}
\newcommand{\thinned}{Y'}
\newcommand{\electrons}{E}
\newcommand{\outermodel}{M^{\uell}}
\newcommand{\innermodel}{M^\mathrm{s}}
\newcommand{\totalmodel}{M^{}}
\newcommand{\state}{X}
\newcommand{\states}{\mathcal{S}}
\newcommand{\prestate}{X'}

\newcommand{\helper}{\beta}

\newcommand{\singlemodels}{\mathcal{F}^\mathrm{s}_{}}
\newcommand{\multimodels}{\mathcal{F}}

\newcommand{\NN}{\mathbb{N}}
\newcommand{\RR}{\mathbb{R}}
\newcommand{\CC}{\mathbb{C}}

\newcommand{\EE}{\mathbb{E}}
\newcommand{\PP}{\mathbb{P}}
\newcommand{\Var}{\mathrm{Var}}

\newcommand{\geob}{q}
\newcommand{\geop}{p}
\newcommand{\poi}{\mu}

\DeclareMathOperator*{\argmax}{arg\,max}

\tcbset{
  defstyle/.style={fonttitle=\bfseries\upshape, fontupper=\slshape,
  arc=0mm, colback=blue!5!white,colframe=blue!75!black},
  theostyle/.style={fonttitle=\bfseries\upshape, fontupper=\slshape,
  colback=red!10!white,colframe=white},
}

\newtcbtheorem{remark}{Remark}{theorem style=plain, 
                               coltitle=red!40!black,
                               colframe=white,
                               colback=red!5!white,
                               breakable}{rem}

\newtcbtheorem{definition}{Definition}{theorem style=plain, 
                           coltitle=green!35!black,
                           colframe=white,
                           colback=green!5!white,
                           breakable}{def}

\newtcbtheorem{lemma}{Lemma}{theorem style=plain, 
                             coltitle=green!35!black,
                             colframe=white,
                             colback=green!5!white,
                             breakable}{lem}

\newtcbtheorem{theorem}{Theorem}{theorem style=plain, 
                             coltitle=green!35!black,
                             colframe=white,
                             colback=green!5!white,
                             breakable}{thm}

\newtcbtheorem{corollary}{Corollary}{theorem style=plain, 
                             coltitle=green!35!black,
                             colframe=white,
                             colback=green!5!white,
                             breakable}{cor}

\title{Statistical Molecule Counting in Super-Resolution Fluorescence
       Microscopy: Towards Quantitative Nanoscopy}

\author{T.\ Staudt\thanks{Institute for Mathematical Stochastics, 
                          Georg-August-University of Göttingen} 
        \and T.\ Aspelmeier\footnotemark[1] 
        \and O.\ Laitenberger\thanks{Laser-Laboratorium Göttingen e.V.}
        \and C.\ Geisler\footnotemark[2] 
        \and A.\ Egner\footnotemark[2] 
        \and A.\ Munk\footnotemark[1]}

\begin{document}
\maketitle
\begin{abstract}
  \noindent Super-resolution microscopy is rapidly gaining importance as an analytical
  tool in the life sciences. A compelling feature is the ability to label
  biological units of interest with fluorescent markers in (living) cells and to
  observe them with considerably higher resolution than conventional microscopy
  permits. The images obtained this way, however, lack an absolute intensity
  scale in terms of numbers of fluorophores observed.
  In this article we discuss state of the art methods to count such fluorophores
  and statistical challenges that come along with it. In particular, we suggest
  a modeling scheme for time series generated by
  single-marker-switching (SMS) microscopy that makes it possible to quantify
  the number of markers in a statistically meaningful manner from the raw data.
  To this end we model the entire process of photon generation in the
  fluorophore, their passage through the microscope, detection and
  photoelectron amplification in the camera, and extraction of time series from
  the microscopic images. At the heart of these modeling steps is a careful
  description of the fluorophore dynamics by a novel hidden Markov model that
  operates on two timescales (HTMM). Besides the fluorophore number, information
  about the kinetic transition rates of the fluorophore's internal states is
  also inferred during estimation. We comment on computational issues that arise
  when applying our model to simulated or measured fluorescence traces and
  illustrate our methodology on simulated data.
\end{abstract}


\section{Introduction}
\label{sec:introduction}

During the past decades cell biology has undergone a profound transition,
shifting its character from qualitative work about basic cell activity to
increasingly quantitative methods to study fine details like the role of
individual proteins for signaling and transport.
This trend was crucially supported by the advancement of super-resolution
microscopy (nanoscopy) techniques, highlighted by the 2014 Nobel prize in
chemistry, which have since become an indispensable tool for modern biomedical
research \cite{hell_1994, betzig_2006, hell_2008, sydor_2015}. While previous
imaging methods for cellular structures were either limited due to a lack of
resolution (like conventional light microscopy) or due to their invasiveness
(like X-ray or electron microscopy), fluorescence nanoscopy enables
high-resolution imaging of living cells to the nanometer scale without the
necessity to prepare samples in ways that prohibit natural biochemical activity.
The limits of super-resolution microscopy, both in principle and application,
are still being explored as progress unfolds at a remarkable pace
\cite{hell_2015,balzarotti2017}.

By now, many initial hurdles for the usage of fluorescence nanoscopy in various
disciplines, like physiology, biology, and medicine, have been overcome:
structures within living prokaryotic and eukaryotic cells are probed on
unprecedented spatial scales in experiments \cite{bakshi2012, laplante2016}, and
popular model organisms like fruit flies and mice are studied in vivo
\cite{schnorrenberg2016, berning2012}.
It is hard to overstate the practical implications of bringing improved imaging
resolution to these fields. From virology \cite{chojnacki2012,
muranyi2013, prescher2015}, immunology \cite{williamson2011, pageon2013} and
neurology \cite{maglione2013, deste2015} to cancer \cite{sharma2012, chen2016}
and plant biology \cite{komis2015}, new ground in fundamental research is
increasingly broken by means of nanoscopy. We exemplarily refer to
\cite{sahl2017} for an in-depth review about the unfolding role of
super-resolution microscopy in cell biology.

The advancement of nanoscopy does not only raise new opportunities for
experimentalists and lab scientists but also for statisticians. They are called
to address a series of challenges that are highly relevant for exploiting the
full potential of state-of-the-art fluorescence microscopy schemes (see
\cite{aspelmeier2015}, where further background on the underlying optics and
the physical modeling of nanoscopy is given from a statistical perspective). 
Indeed, all current implementations of super-resolution microscopy are affected
by the inherently stochastic behavior of fluorescent molecules, or fluorophores,
which (randomly) emit photons if struck by incident light.
In modes of nanoscopy that operate in a coordinate-targeted way
(scanning observation points), like STED \cite{hell_1994, hell_2008} and RESOLFT
\cite{hofmann2005, brakemann2011, grotjohann2011}, this stochasticity often
plays a secondary role.
Still, for quantitative analysis of the images (i.e., counting the actual number
of fluorophores), the photon emission statistics turns out to be central. For
example, individual fluorophores can be identified in STED nanoscopy by
measuring the simultaneous arrival of emitted photons \cite{ta2015, koenig2019}.
In case of RESOLFT, an on-off Markov model for fluorophores has recently been
demonstrated to be capable of extracting the contribution of single fluorophores
in the total signal \cite{frahm2019}.
Other methodologies, like MINFLUX \cite{balzarotti2017, eilers2018}, rely on
a statistical treatment by design. In MINFLUX -- fluorescence nanoscopy via
minimal photon fluxes -- a doughnut shaped laser intensity profile targeted to
different spots on a biological sample is used to excite a fluorophore with
unknown position.
Based on the (approximately Poisson distributed) number of photons measured as
response for each position of the excitation spot, the location of the
fluorophore is inferred statistically, e.g., via maximum likelihood estimation.
Questions regarding the optimal measurement design -- where to place the spots
and which laser profile to use -- naturally fit a Bayesian perspective and are
still open for investigation.

\begin{figure}
  \centering
  \includegraphics[width=0.95\textwidth]{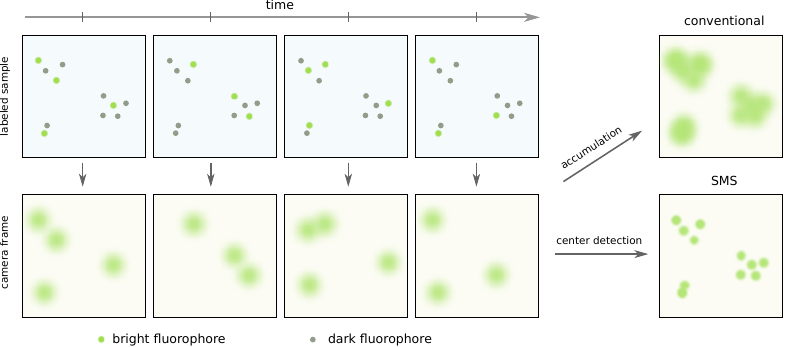}
  \caption{Principle of single-marker-switching microscopy. By exciting
  a biological sample that is labeled with fluorophores (top row) via a suitable
  laser, a temporal series of frames capturing fluorescent activity is recorded
  (bottom row). In each frame, only a sparse selection of fluorophores emits
  photons (green circles). The recorded images are blurry because of inevitable
  diffraction effects. Still, the center positions of the individual diffraction
  limited spots can be determined with higher precision due to spatial sparsity.
  This can be used to create a pointillistic nanoscopy image with superior
  resolution as compared to conventional fluorescence microscopy, where the
  photons emitted by all fluorophores would be recorded at the same time.
  }
  \label{fig:sms}
\end{figure}

The major focus of this article, however, will be another family of nanoscopy
schemes, which exploit the fact that fluorophores have a tendency to blink over time,
meaning that they (randomly) switch between active and inactive states. Under
suitable conditions, fluorophores can thus be observed and localized
individually even when clustered together. Methods that make use of this
switching property are collected under the umbrella term single-marker-switching
(SMS) nanoscopy, and include PALM \cite{betzig_2006}, STORM
\cite{rust_2006}, GSDIM \cite{folling_2008}, or variations thereof
\cite{hess_2006, egner_2007, van_2011}.
SMS nanoscopy works by recording a series of diffraction limited fluorescence
images (or frames) in which only a small number of fluorophores is active and
emits photons during the respective exposure. Spatially close molecules are
therefore likely to be separated in time. As illustrated in Figure
\ref{fig:sms}, the detected fluorophore positions from all frames can be used to
create a pointillistic image with superior resolution. For a video that compares
conventional fluorescence microscopy and SMS nanoscoy of a Rhodamine labeled
microtubular network on the basis of $30\,000$ frames of experimental data, see
\url{{http://stochastik.math.uni-goettingen.de/SMSData}}.

The stochastic nature of the frames recorded during SMS microscopy opens up
a rich and fruitful field for statistical investigation. Indeed, major issues
like the correction of spatial drifts in the image sequence, which
originally required experimental intervention via so-called fiducial markers,
have recently been tackled by fully statistical means -- see \cite{hartmann2016}
and the references therein.
Another emergent topic is the temporal statistical modeling of the fluorophore
dynamics, which also plays a crucial role for the present article. The most
prominent approaches in this context are (hidden) Markov models, see
\cite{messina_2006, rollins_stochastic_2015, tsekouras_2016, hummer2016}.
Recently, in \cite{patel2019}, the photo switching behavior
of fluorophores was characterized by a specifically tailored time-homogeneous
hidden Markov model that reliably improves the estimation of kinetic transition
rates from SMS data when compared to more basic methods, like exponential
fitting of the dwell-times of the fluorophore in active/inactive states
\cite{lin2015}.

Moreover, Markov models allow for the refined extraction of quantitative
information from SMS images, like counting the number of individual fluorophores
in given image regions. This task of \enquote{quantitative nanoscopy} turns out
to be much more involved than it appears at a first glance.
The difficulty is that each fluorophore leaves a sophisticated intensity trace
on the recorded image series, as it only causes a visible spot during frames
in which the fluorophore is active -- else it is invisible.
Consequently, when a spot on the microscopic frames is lit up several times
consecutively, it is not evident how many (close-by) fluorophores are
responsible for the observed intensity pattern. This problem of mapping
fluorescence intensity traces to the number of contributing molecules 
is of high practical relevance and poses a fundamental challenge for the
application of SMS microscopy in quantitative biology.

In recent years, several methods to obtain such fluorophore numbers from
fluorescence images have been proposed \cite{lee_counting_2012,
rollins_stochastic_2015, tsekouras_2016, hummer2016}. They usually rely on the
detection of switching events or on counting the number of steps during
photobleaching (i.e., a fluorophore becoming irreversibly inactive). While these
methods have been successfully applied to count 50 fluorophores and more in
specific circumstances \cite{lee_counting_2012, tsekouras_2016}, they can be
prone to errors when misidentifying switching events or bleaching steps. These
issues are particularly detrimental in the presence of many fluorophores within
a diffraction limited region or when the fluorophore kinetics of bleaching and
switching are fast in comparison to the image acquisition rate.

In this article, we lay the statistical foundations for a new method to estimate
the number of fluorophores on SMS nanoscopy images introduced in
\cite{laitenberger_2018}.
Contrary to established methods, no step identification -- which usually
involves the choice of fluorescence levels or rate thresholds and depends on
bleaching or switching -- is necessary. This becomes possible by the careful
statistical modeling and analysis of the whole imaging process: from photon
generation in the fluorophore to signal amplification in the CCD camera. Our
approach makes use of the full history of the recorded intensity information and
exploits temporal correlations in the signal.
The core component of the model is an accurate description of the fluorophore
behavior in terms of a novel hidden Markov model that operates on two distinct
timescales. It separates the fast dynamics that govern the emission of single
photons during the exposure from the slow dynamics that describe fluorophore
kinetics for states with dwell times longer than the exposure for a single
frame. Although our Markov model is time-inhomogeneous, estimation of the
fluorophore number and other kinetic parameters can be performed by applying the
maximum likelihood principle to a simplified expression of the model's total
likelihood. This simplification is based on a second-order approximation to the
true likelihood and is derived by exploiting spectral properties of the model.
Intriguingly, the inference takes place in an unusual setting: the quantity we
want to estimate -- the fluorophore number -- is a feature of the initial state
of the model and is lost in the asymptotic behavior for long times due to
bleaching. In \cite{laitenberger_2018}, the method has been experimentally
verified on super-resolution images of DNA origami structures. This will be
complemented by simulation results in the present work.

The article is organized as follows. In Section \ref{sec:modelingsteps}, we
provide an overview of the single modeling steps that contribute to our total
model for the fluorescent time traces, and we briefly describe how we estimate
the fluorophore number with it.
Section \ref{sec:fluorophore} contains a detailed treatment of the
fluorophore dynamics. We formulate the hidden two-timescale Markov model
(abbreviated by HTMM) that is based on the description of fluorescent molecules
as Markov chains, acting on different timescales with different
transition rates. 
In particular, we derive expressions for the expectation and (co-)variance of
the number of emitted photons in each frame, and provide results about spectral
properties of the transition matrix which are useful for computational purposes
(Appendix \ref{sec:diag} and \ref{sec:eig}).
In Section \ref{sec:thinning}, we investigate how the number of emitted photons
is transformed through (i) statistical thinning in the microscope and (ii)
processing and amplification in the detector. To assist readability, the
central notation that is introduced in sections \ref{sec:fluorophore} and
\ref{sec:thinning} is surveyed in Table \ref{tab:notation} on page
\pageref{tab:notation}.
In Section \ref{sec:estimation}, we then introduce the simplified
\enquote{pseudo log-likelihood}, comment on numerical issues for maximum
likelihood estimation based on it, and present estimation results for simulated
fluorescence intensity traces. Finally, in Section \ref{sec:alexa}, we
specialize our general model to the commonly used fluorophore Alexa 647. 
Section \ref{sec:outlook} contains a brief outlook that emphasizes open questions
related to our work.

\section{Modeling and Estimation}
\label{sec:modelingsteps}
Super-resolution microscopy with single marker switching (SMS) relies on
a series of fluorescence microscopy images, or frames, with only a small
fraction of active fluorophores per image. This way, spatially close
fluorophores are separated in time since they are unlikely to emit photons
simultaneously. The resulting frames are used to localize the marker molecules
with a superior precision on the nanometer scale \cite{sydor_2015}.
The imaging is affected by the quantum physical behavior of the fluorophore,
which leads to \emph{switching} and \emph{bleaching}, and by a series of
subsequent manipulations of the emitted photons until they are detected by the
camera and transformed to digital values
\cite{aspelmeier2015}. Each step in this chain (depicted in Figure
\ref{fig:steps}) modifies the original signal -- photons emitted by the
fluorophore -- in a characteristic way and has to be taken into account. In the
following, we present an outline of our approach to estimate the fluorophore
number based on time traces extracted from a series of $T$ camera images. More
detailed considerations follow in subsequent sections.

\begin{figure}[bt!]
  \centering
  \includegraphics[width=0.9\textwidth]{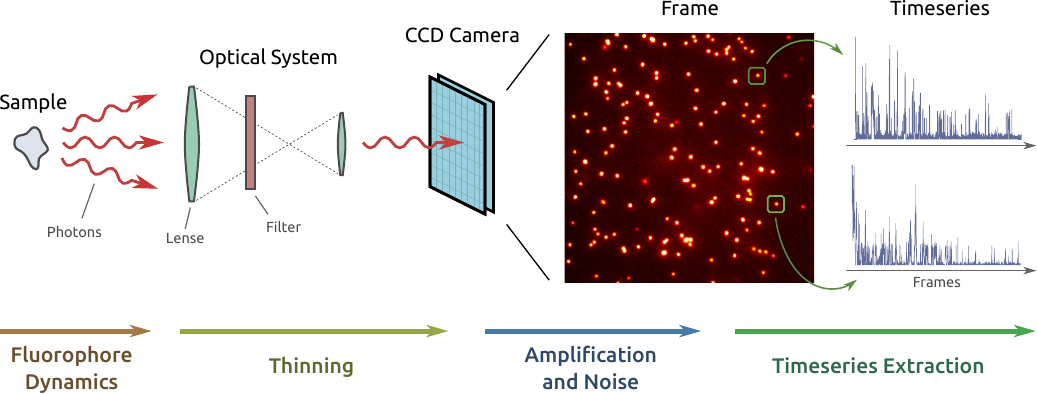}
  \caption{Overview of the modeling steps. During exposure, fluorophores within
    a labeled biological sample emit photons due to laser excitation.
    With a certain probability, these photons pass
    through the microscope (optical system) and are registered by a CCD camera.
    Over the course of the experiment, a series of camera frames is recorded.
    Summing up the intensities over a region of interest (ROI, green
    boxes) for each image yields a time series that captures the fluorescent
    activity in the respective ROI. 
  }
  \label{fig:steps}
\end{figure}

\paragraph{Single fluorophore.} 
Fluorophore dynamics is successfully modeled by Markov
chains \cite{messina_2006, patel2019}. The states of these chains roughly
correspond to quantum physical states of the molecule (see Remark
\ref{rem:compoundstates} below), which can exhibit very diverse lifetimes.
The respective transition rates are governed by quantum mechanical kinetics
that sensitively depend on the biochemical properties of the fluorophore's
neighbourhood in the sample. Two of the states have a distinguished role in our model:
the \emph{bright} state, in which absorption and emission of photons is
possible, and the \emph{bleached} state, in which dyes have
irreversibly lost their fluorescence functionality.  Additionally, a number
of temporary \emph{dark} states, which, e.g., correspond to triplet or redox
states of the fluorophore \cite{vogelsang2010}, are usually necessary for
a faithful description.
\begin{remark}{}{compoundstates}
  In Markov chains, states with the same transition rates can be combined into
  a single state without losing the Markov property (see Appendix
  \ref{sec:lumped} for details).
  A reasonable fluorophore model does therefore not have to include every
  possible quantum physical state explicitly (like fine-structured rotational
  and vibrational substates). Rather, it only has to capture classes of states
  with similar dwell times and transition behavior. The number of such classes
  can be estimated from the data.
\end{remark}

The phenomenon of fluorophores jumping between the bright and temporary dark
states is denoted as \emph{blinking} or \emph{switching}.  In our generic model
for fluorescence, we finely resolve the fast dynamics inherent to the bright
state, like single photon emissions, and model it as a Markov chain in its own
right.  This gives rise to a description that operates on two different time
scales:
a fast \emph{inner model} that runs during the exposure time, and a slow
\emph{outer model} that captures states that are expected to persist over
several frames. Figure \ref{fig:model-alexa} depicts our choice of states for
the fluorophore Alexa 647, which we investigate more detailed in Section
\ref{sec:alexa}. 

\begin{figure}[t!b]
  \centering
  \includegraphics[width=0.635\linewidth]{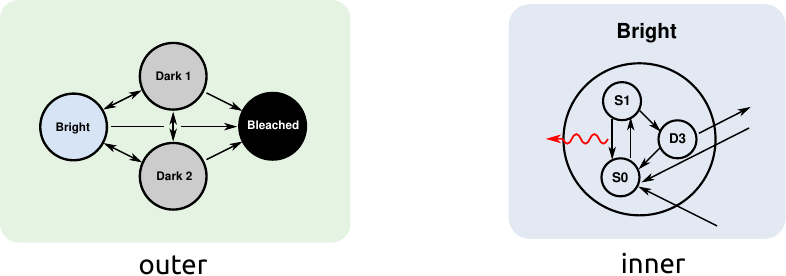}
  \vspace{0.175cm}
  \caption{Exemplary inner and outer models that are used to describe the
    dynamics of the fluorophore Alexa 647 (see \cite{laitenberger_2018}). The
    wiggling red arrow indicates that transitions from the singlet state
    $\mathrm{S}_1$ to the ground state $\mathrm{S}_0$ cause the emission of
    a photon. The state $\mathrm{D}_3$ is a short-lived dark state.}
  \label{fig:model-alexa}
\end{figure}

\begin{figure}
  \centering
  \vspace{0.6cm}
  \includegraphics[width=0.735\linewidth]{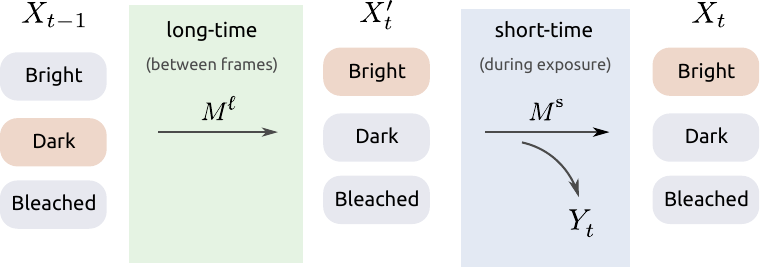}
  \vspace{0.15cm}
  \caption{Single fluorophore model. At time $t-1$, the
    fluorophore can be in one of several distinct outer states $\state_{t-1}$.
    When transitioning from $\state_{t-1}$ to $\state_t$, we apply the
    long-time matrix $\outermodel$ to describe the dynamics that takes
    place between separate frames, and then apply the short-time matrix
    $\innermodel$ for the effects of the fast dynamics on the outer state
    during exposure. In the latter step, the fluorophore emits
    a number $Y_t$ of photons in frame $t$.}
  \label{fig:model-timeline}
\end{figure}

Let $\prestate_t$ denote the
outer state of the fluorophore directly before the $t$-th exposure, and let
$\state_t$ denote its state after the exposure. The transition from
$\state_{t-1}$ to $\state_t$ is depicted in Figure \ref{fig:model-timeline}.
It is modeled through applying one step of the outer dynamics on
$\state_{t-1}$, which yields $\prestate_t$, and then running the inner
model, which changes $\prestate_t$ to $\state_t$, and which also yields
a number $\photons_t$ of emitted photons during frame $t$. The distribution
of $\photons_t$ depends on both $\prestate_t$ and $\state_t$. A complete
description of our model for fluorophore dynamics is therefore given by
a transition matrix $\outermodel$ for the long-time step, a matrix
$\innermodel$ for the short-time step, and the distribution $p_{xx'}$ of
$\photons_t$ conditioned on $\state_t = x$ and $\prestate_t = x'$, which we
assume to be stationary. 
The combined chain 
\begin{equation*}
  \big(\state_0, \prestate_1, \state_1, \prestate_2, \ldots\big)
\end{equation*}
of outer states is an inhomogeneous Markov chain with alternating transition
matrices $\outermodel$ and $\innermodel$, while the individual chains
$(\state_t)_t$ and $(\prestate_t)_t$ are homogeneous with transition matrices
$\innermodel\outermodel$ and $\outermodel\innermodel$, respectively.

In experiment, the states $\state_t$ and $\prestate_t$ cannot be observed
directly. We only obtain outputs of the measurement device (e.g., a CCD camera)
generated through the $Y_t$ emitted photons. This makes our ansatz a hidden
Markov model. 
In Section \ref{sec:fluorophore}, we derive the generating function of the
process $Y = (\photons_t)_{t=1}^T$ and obtain its expectation $\mu$ and the
covariance $\Sigma$, which are eventually used to estimate the number of
fluorophores in Section \ref{sec:estimation}.

\begin{remark}{notation}{notation}
  We refer to the model outlined above as hidden two-timescale Markov model, or
  HTMM. The observable part of this model, $Y_t$, denotes the number of photons
  that are emitted in the time interval between $\prestate_t$ and $\state_t$.
  One can therefore think of $\prestate_t$ as the state $\state_{t^-}$ directly
  prior to $\state_t$, and $Y_t$ as an observation that accumulates from $t^-$
  to $t$. 
\end{remark}

\paragraph{Microscope and camera.}
Photons emitted by fluorescent dyes are directed randomly and may fail to enter
the microscope, such that they are lost for the experiment. In addition,
a photon may be absorbed by lenses, filters, or mirrors within the
optical path.
Consequently, each emitted photon has a probability $p_\mathrm{c} < 1$ to reach
the camera. When it reaches the camera, the position of the photon on the CCD
sensor is randomized due to diffraction: light originating from
a point source is spread to a blurred spot on the detector interface. 
From the viewpoint of classical physics, where light is modeled as a wave of
electromagnetic radiation, this blurring is described by a convolution of the
light intensity distribution with a nonnegative point spread function $h$ (see
\cite{goodman_introduction_1996, born_principles_1999} for the underlying
physics and \cite{aspelmeier2015} for a treatment in the context of
statistics). 
In the quantum mechanical interpretation of light as photons,
$h(z)$ denotes the probability that a photon emitted at the origin of the sample
incides at pixel $z$ on the detector, which leads to a multinomial distribution
of incident photons to pixel locations. When the photon arrives at a pixel $z$,
it is absorbed with a certain probability $p_\mathrm{a}$ and a so-called photo
electron, i.e., an electron ejected from the detector material due to energy
transfer from the photon, is released. The total chance for a photon to reach
the detector at any pixel and be absorbed is denoted by $p_\mathrm{d}
= p_\mathrm{c} \, p_\mathrm{a}$.

We call a region $R$ on the image that captures the blurred spot created by
one (or several close-by) fluorophores a region of interest (ROI).
The total number $\thinned_t$ of detected photons in $R$ is given by
\begin{equation*}
  \thinned_t 
    = \sum_{z\in R}^{} \thinned_{t,z} 
\end{equation*}
where $\thinned_{t,z}$ is the number of photons detected at pixel $z\in R$.
Since we assume that the electrical circuits underlying individual pixels are
identical in their properties, we can ignore the spatial distribution of
photons within one ROI and work with $\thinned_t ~\sim~ \mathrm{Bin}(\photons_t,
p_\mathrm{d})$ directly. This amounts to a binomial thinning of $\photons_t$
\cite{harremoes2010}.

Since the electrical charge of a photoelectron is too small to be
detected reliably, cameras employ an electron multiplying system that
operates stochastically \cite{robbins_noise_2003,hirsch_stochastic_2013}.
Let $\mathcal{D}$ denote the distribution for the number of electrons after
amplification of the incoming electron in the CCD. Then the final camera
output value $\ccdvals_t$, when summed over $R$, is given by 
\begin{equation} \label{eq:camera}
  \ccdvals_{t} = c\sum_{k=1}^{\thinned_t} U_{t,k} + \epsilon_t + o, 
\end{equation}
with $U_{t,k}\sim\mathcal{D}$ i.i.d.\ for
all $t$ and $k$. The constant factor $c > 0$ results from the analog-to-digital
conversion of the accumulated electron charge in the pixels, and the random
variables $\epsilon_t$ collect different contributions of inevitable additional
randomness -- like background photons, thermal electrons in the electronics, or
readout noise. Additionally, a constant positive offset $o$ is added to the
camera output to avoid noise induced fluctuations into the negative domain.

\paragraph{Multiple fluorophores.} 
Each fluorophore in the bright state produces a diffraction-limited spot during
exposure, as seen in the frame shown in Figure \ref{fig:steps}.
The major difficulty for estimating the number of fluorophores reliably results
from the fact that several fluorophores can contribute to the same spot if their
mutual distance is small and if they are bright simultaneously. 
The core contribution of this article is to use the information from a temporal series
of frames to estimate the total number $m$ of fluorophores that are present in
a given region of interest $R$.
\begin{remark}{}{quantity_of_interest}
  The unknown number $m$ of fluorophores is the major quantity of interest to be
  estimated in a single ROI. By combining estimates for $m$ from different ROIs,
  one can obtain quantitative information on the spatial fluorophore density in
  the entire image. 
\end{remark}
A crucial assumption that we use to model multi-fluorophore systems is
\emph{statistical independence}, i.e., that no (relevant) physical interactions
between the single fluorophores take place. We also assume that all $m$
fluorophores are \emph{identical} in their physical behavior, meaning that they
can be described by a common fluorophore model with a common set of parameters.
Then, the total number of emitted photons is given by the sum of $m$ independent
copies $\photons^1_t, \ldots, \photons^m_t$ of the process $\photons_t$,
\begin{equation}\label{eq:sum_over_R}
  \photons^{(m)}_t = \sum_{k=1}^{m} \photons^k_t.
\end{equation}
Similarly, the time series obtained by summing the CCD values over the region
$R$ is composed of $m$ independent versions $\ccdvals_{t}^1, \dots,
\ccdvals_{t}^m$ of
$\ccdvals_{t}$. Therefore, the total signal we observe is encoded in the
process
\begin{equation*}
  \ccdvals^{(m)}_t = \sum_{k=1}^{m} \ccdvals_{t}^k.
\end{equation*}

\begin{remark}{}{iid_assumptions}
The assumptions of independence and identical distribution are approximations
that are justified for many typical experimental situations.  Still, they can be
violated, e.g., if the spatial distance of neighbouring fluorophores is very
small ($<10\,\mathrm{nm}$). Then, interactions like FRET (Förster Resonance
Energy Transmission) become likely. The experimental study
\cite{laitenberger_2018} highlights that our model indeed produces inconsistent
results in this case. 
\end{remark}

\paragraph{Estimation.} 
Our objective is to estimate $m$ from a realization $y$ of the process
$\ccdvals^{(m)} = \big(\ccdvals^{(m)}_t\big)_t$. Besides $m$, there are several
other parameters that may have to be estimated, like the transition rates in the
HTMM, or the initial distribution of the outer state. These
parameters depend sensitively on details of the experimental setting, like the
fluorophore type, the biochemical conditions in the sample, or the applied laser
wavelengths and intensities. Some of these properties may vary from ROI to ROI.
Furthermore, different types of fluorophores may even require different inner
or outer models. This poses an interesting problem in model selection -- which
we will, however, not address in this article.

Since the number of unknown parameters, which we call $\gamma$ for
the moment, is typically small (e.g., at most 16 for our Alexa 647 model with
three dark states, see Section \ref{sec:alexa}), it is near at hand to employ
maximum likelihood estimation.
However, the computation of the MLE requires
that we can evaluate the log-likelihood $l_y(\gamma)$ of the full model, which
is unfeasible for two reasons: first, the number of terms in $l_y(\gamma)$
turns out to be overwhelming even for a moderate number $T$ of frames; and 
secondly, we lack information about how the signal is transformed in the camera,
since manufacturers usually only provide information about the first two moments
of the camera-statistics $\mathcal{D}$.
We therefore choose an approach that is based on approximating $\ccdvals^{(m)}$ by
a Gaussian process with the same expectation $\mu = \mu(\gamma)$ and covariance
$\Sigma = \Sigma(\gamma)$ as
$\ccdvals^{(m)}$. This leads to the \emph{pseudo log-likelihood}
\begin{equation*}
  \tilde{l}_y(\gamma) = -\frac{1}{2}\big[(y - \mu) \,
  \Sigma^{-1}\, (y - \mu) + \log\det\Sigma \big],
\end{equation*}
and parameter estimation reduces to maximizing $\tilde{l}_y(\gamma)$,
which is still challenging but becomes numerically feasible. In particular, we
have fewer degrees of freedom: for Alexa 647, only 11 (compared to 16)
independent parameters are necessary to fully describe $\mu$ and $\Sigma$, since
the first two moments do not rely on all transition probabilities of the
HTMM individually.
In Section \ref{sec:estimation}, we address a number of subtleties that come
along with this approach, like nonlinear constraints on the parameter space
and practical complications with the numerical optimization. 
As a proof-of-concept of our approach, we also include exemplary
estimation results on simulated data.

We stress that other methods of estimation are certainly of interest, too. In
particular, a Bayesian approach becomes feasible when prior knowledge on
parameters is available. We do, however, not pursue this issue further in this
article and contain ourselves to (pseudo) maximum likelihood based statistical
analysis.

\section{Fluorophore Dynamics}
\label{sec:fluorophore}
The dynamics of fluorophores is at the heart of our model for the imaging
process in Figure \ref{fig:steps}.
Due to the i.i.d.\ assumption when modeling multiple fluorophores (see Remark
\ref{rem:iid_assumptions}), this effectively amounts to modeling a single
fluorophore.
In this section, we will treat the short-time dynamics as a ``black box'' with
as few assumptions as is necessary for our intentions. Only later, when we
specialize the model to the fluorophore Alexa 647 in Section \ref{sec:alexa}, we
elaborate in detail on a concrete inner model. 
As a guidance for the derivations that follow, consulting Table
\ref{tab:notation} on page \pageref{tab:notation} might prove helpful, as it
summarizes the relevant notation that is introduced in this and the next
section.

\paragraph{Fluorophore model.}
The outer state space of a fluorophore is described by one bright and $r \in\NN$
dark states, including the bleached one. See Figure \ref{fig:model-alexa} for an
example.
We name the state space $\states$ and denote its elements by $x\in\states
= \{0, \dots, r\}$, with $x=0$ being the bright and $x=r$ being the bleached
state. On $\states$, we consider two coupled time-discrete Markov chains
$(\state_t)_{t=0}^T$\label{state} and $(\prestate_t)_{t=1}^T$\label{prestate},
where $T\in\NN$ is the number of frames. 
The evolution of $\state_t$ and $\prestate_t$ is given by the (stationary)
transition matrices $\outermodel$ and $\innermodel$, where
\begin{subequations}
\begin{equation}\label{eq:deftransition}
  \outermodel_{x'x} = \PP\big(\prestate_t = x'\,\big|\,\state_{t-1} = x\big)
  \qquad\text{and}\qquad
  \innermodel_{xx'} = \PP\big(\state_t = x\,\big|\,\prestate_t = x'\big)
\end{equation}
for $x,x'\in\states$.
We interpret $\prestate_t$ as state of the fluorophore directly before exposure
in frame $t$, and $\state_t$ as state directly after exposure (see Remark
\ref{rem:notation} and Figure \ref{fig:model-timeline}).
The full transition matrix for $\state_t$ is given by
\begin{equation}\label{eq:deftotaltransition}
  \totalmodel = \innermodel\outermodel.
\end{equation}%
\label{eq:defalltransition}%
\end{subequations}%
The transition of the fluorophore from $X'_t$ to $X_t$ during exposure is
governed by the inner model, which also determines the number
$\photons_t$\label{photons} of photons that are emitted in the corresponding
frame. We characterize the photon statistics of the inner model by the
conditional distributions
\begin{equation}\label{eq:photonstatistics}
  p_{xx'}(y) 
    = \PP\big(\photons_t = y\,\big|\,\state_t 
    = x, \prestate_t = x'\big)
\end{equation}
for $y\in\mathbb{N}_0$, which we assume to be time-stationary. The probabilities
in \eqref{eq:photonstatistics} are collected in the matrix $P(y)
= (p_{xx'})_{x,x'\in\states}$. We furthermore use the symbol $\nu
= (\nu_x)_{x\in\states}$\label{initdist} to denote the initial distribution, i.e., the
distribution of $X_0$.  In total, specification of $\innermodel$, $\outermodel$,
$P$, and $\nu$ completely defines a single-fluorophore model. 

\begin{definition}{}{htmm}
  Any observable process $Y = (Y_t)_t$ as constructed above, with conditional
  distributions $P$ in \eqref{eq:photonstatistics}, transition matrices $\innermodel$
  and $\outermodel$ defined in \eqref{eq:defalltransition}, and initial (hidden)
  distribution $\nu$, is denoted as hidden two-timescale
  Markov model, or HTMM.
\end{definition}

\begin{lemma}{}{likelihood}
  The likelihood of an HTMM under observation of a time series $y
  = (y_t)_{t=1}^T$ is 
  \begin{equation}
  l_y(\innermodel, \outermodel, P, \nu) 
  = \sum_{x\in\states^{T+1}}\left(\prod_{t=1}^T\, 
    \sum_{x'\in\states} p_{x_{t}x'}(y_t)\innermodel_{x_{t}x'}
    \outermodel_{x'x_{t-1}}\right)\nu_{x_0},
  \label{eq:likelihood}
  \end{equation}
where the outer sum covers all tuples $x = (x_0, x_1, \dots, x_T)
\in\states^{T+1}$. 
\end{lemma}
\begin{proof}
A single transition step under observation of $\photons_t = y$ photons is
described by
\begin{equation*}
  \PP\big(\photons_t = y, \state_t = x \,\big|\,\state_{t-1} = z\big)
    = \sum_{x'\in\states} p_{xx'}(y)\innermodel_{xx'}\outermodel_{x'z},
\end{equation*}
and the probability to observe the full time series $(y_t)_{t=1}^T$ can be
written as
\begin{equation*}
  \PP\big(\photons_t = y_t ~\text{for all}~t\big) 
  = \sum_{x\in\states^{T+1}}\left(\prod_{t=1}^T\, 
    \PP\big(\photons_t = y_t, X_t = x_t\,\big|\,X_{t-1} 
   = x_{t-1}\big)\right)\nu_{x_0} \nonumber.
\end{equation*}
Combining these two equations yields the stated result.
\end{proof}

\begin{remark}{}{complexity}
  The outer sum in equation \eqref{eq:likelihood} contains $(r+1)^{T+1}$ terms,
  and we are not aware of a way to significantly simplify the expression
  under general circumstances. Since there is at least one dark outer state
  (meaning $r \ge 1$) and the number $T$ of frames for SMS microscopy is often
  well above 1000, it is impossible to directly evaluate the likelihood $l_y$
  numerically ($2^{1000} > 10^{300}$).
\end{remark}

\begin{remark}{}{}
We will eventually extend the HTMM in Definition \ref{def:htmm} by (i) the
generalization to multiple i.i.d.\ fluorophores, and (ii) the additional
statistical modeling of the imaging process (see Section \ref{sec:thinning}).
These enhanced models will for convenience also be referred to as HTMMs, since
the context usually clarifies which specific fluorophore model is meant.
\end{remark}

\paragraph{Model restrictions.} 
We introduce several restrictions on our model in order to reflect physical
fluorophore properties and to make the analysis of \eqref{eq:likelihood} viable.
The most evident constraint is that the bleached state
$x = r$ acts as an absorbing state for both $\innermodel$ and
$\outermodel$. We also assume that the fluorophore can leave the bright
state $x = 0$ only during application of the inner model.
Conversely, we suppose that a fluorophore that is not in its bright state is
unaffected by the inner model. 
With these restrictions in place, $\outermodel$ and $\innermodel$ can be
brought in the respective parametric forms
\begin{equation}\label{eq:modelmatrices}
  \outermodel = \begin{pmatrix}
    1 & q_{01} & \cdots & q_{0(r-1)} & 0 \\
    0 & q_{11} & \cdots & q_{1(r-1)} & 0 \\
    \vdots & \vdots & \ddots & \vdots & \vdots\\
    0 & q_{(r-1)1} & \cdots & q_{(r-1)(r-1)} & 0 \\
    0 & q_{r1} & \cdots & q_{r(r-1)} & 1 \\
  \end{pmatrix}
  \quad~~\text{and}\quad~~
  \innermodel = \begin{pmatrix}
    q_{00} & 0 & \cdots & 0 & 0 \\
    q_{10} & 1 & \cdots & 0 & 0 \\
    \vdots & \vdots & \ddots & \vdots & \vdots\\
    q_{(r-1)0} & 0& \cdots &  1 & 0 \\
    q_{r0} & 0& \cdots &  0 & 1 \\
  \end{pmatrix},
\end{equation}
where we defined transition probabilities $q_{xz}\in[0, 1]$ for $x\in\states$
and $z\in\states\smallsetminus \{r\}$. The restrictions also imply that no
photons are emitted if the exposure starts in a non-bright state, meaning that
$p_{xx'}(0) = 1$ for $x'\neq 0$. In contrast, if the exposure begins in the
bright state $x' = 0$, the fluorophore will produce photons and may switch
to any other state $x\in\states$ until the end of the exposure. 

An additional assumption that is required to make the model analytically
tractable is that the distribution $p_{x0}$ does not depend on the final state
$x$ if the fluorophore exits the bright state. This means that
\begin{equation}\label{eq:commonexit}
  p_{x0} = p_{10} \quad\text{if}\quad x > 0.
\end{equation}

\begin{remark}{}{commonexit}
Condition \eqref{eq:commonexit} can be understood as assuming a common
\enquote{exit state} in the inner model that is the only possibility for the
fluorophore to become dark during exposure. From this state, it can then jump to
all dark states of the outer model as soon as the frame ends. Note that this
exit state does not have to correspond to a (single) physical state: it could
cover several states with similar exit conditions (see Remark
\ref{rem:compoundstates} in this context).
\end{remark}

Under these constraints, the conditional distributions of the photon statistics
defined in \eqref{eq:photonstatistics} read
\begin{equation}
  \label{eq:great_p}
  P(y) = \begin{pmatrix}
    p_{00}(y) & \delta(y) & \cdots & \delta(y) \\
    p_{10}(y) & \delta(y) & \cdots & \delta(y) \\
    \vdots    & \vdots    & \ddots & \vdots    \\ 
    p_{10}(y) & \delta(y) & \cdots & \delta(y) \\
  \end{pmatrix},
\end{equation}
where $\delta$ denotes the Dirac measure with point mass $1$ on $y = 0$.
Fluorophore models that satisfy conditions \eqref{eq:modelmatrices} and
\eqref{eq:great_p} are collected in the set $\singlemodels$
of (physical) single-fluorophore models.

\paragraph{Generating function.} 
Even when exploiting the additional constraints for $\singlemodels$, the process
$\photons = (\photons_t)_t$ of emitted photons remains very complex. In
particular, it is hard to use the HTMM for straightforward inference. Maximum
likelihood estimation, e.g., is impossible for real world datasets due to
the prohibitive expense of calculating the likelihood, see Remark
\ref{rem:complexity}.  We can, however, use the moment generating function of
$\photons$ and the specific structure of $\singlemodels$ to calculate
expressions for the expectations $\mu_t = \EE[ \photons_t]$ and the covariance
$\Sigma$,
\begin{equation*}
  \Sigma_{tt'} = \EE[(\photons_t - \mu_t) (\photons_{t'} - \mu_{t'})],
\end{equation*}
where $t,t' = 1, \dots, T$. These moments carry relevant information and allow
recovering the number of fluorophores in case of multiple molecules (see Section
\ref{sec:estimation}). For preparation, we first look at the moment generating
matrix $G(s)$ associated to $P(y)$ and find
\begin{equation}\label{eq:great_g}
  G(s) = \begin{pmatrix}
    G_{00}(s) & 1 & \cdots & 1 \\
    G_{10}(s) & 1 & \cdots & 1 \\
    \vdots & \vdots & \ddots & \vdots \\
    G_{10}(s) & 1 & \cdots & 1 \\
  \end{pmatrix},
\end{equation}
where $G_{00}(s) = \EE\big[e^{s\photons_t} \,|\,\state_t = \prestate_t = 0\big]$
and $G_{10}(s) = \EE\big[e^{s\photons_t} \,|\,\state_t = 1, \prestate_t = 0\big]$.
Then, we define the auxiliary matrix 
\begin{equation}\label{eq:matrix_H}
  H(s) = \big(G(s) \circ \innermodel\big) \outermodel,
\end{equation}
where $\circ$ denotes the entry-wise (Hadamard) product. In the following, we
will only consider inner models for which the expectations $G_{00}$ and
$G_{10}$ exist and are finite in some vicinity $(-\epsilon, \epsilon)$ of zero
for $\epsilon > 0$. In particular, this implies that all derivatives of $H$
exist at $s = 0$ \cite{feller2008}.
\begin{lemma}{}{momentgenerating}
  The moment generating function $G_Y$ of the process $Y$ is
  \begin{equation} \label{eq:generating_function}
    G_\photons(\tau) = (1, \dots, 1) \,H(\tau_T)\cdots H(\tau_1)\,\nu
  \end{equation}
  for $\tau = (\tau_1, \dots, \tau_T) \in(-\epsilon,\epsilon)^T$.
\end{lemma}
\begin{proof}
  First, we note that the matrix $H(s) = \big(G(s)\circ \innermodel\big)
  \outermodel$ has components
  \begin{equation}\label{eq:helper_H}
    H_{xz}(s) 
      = \sum_{x'\in\states} G_{xx'}(s) \innermodel_{xx'}\outermodel_{x'z}
      = \sum_{x'\in\states} \sum_{y\in\mathbb{N}_0} p_{xx'}(y) e^{s y} \innermodel_{xx'}\outermodel_{x'z}
  \end{equation}
  for $x,z\in\states$ and $-\epsilon < s < \epsilon$. Next, consulting
  result \eqref{eq:likelihood} of Lemma \ref{lem:likelihood} and using the
  notation $y = (y_1,\ldots, y_T) \in \mathbb{N}_0^T$ we can write 
\begin{align*}
  G_\photons(\tau) 
    &= \EE \left[\exp\left(\sum_{t=1}^T \tau_t \photons_t\right)\right] \\
    &= \sum_{y\in\NN_0^T} \PP\big(Y_t = y_t~\text{for all}~t\big) \prod_{t=1}^Te^{\tau_t y_t}\\
    &= \sum_{y\in\NN_0^T} \sum_{x\in\states^{T+1}}\left(\prod_{t=1}^T\, \sum_{x'\in\states} p_{x_{t}x'}(y_t)e^{\tau_t y_t}\innermodel_{x_{t}x'}\outermodel_{x'x_{t-1}}\right)\nu_{x_0} \\
    &= \sum_{x\in\states^{T+1}}\left(\prod_{t=1}^T\sum_{x'\in\states}\sum_{y_t\in\NN_0} p_{x_{t}x'}(y_t)e^{\tau_t y_t}\innermodel_{x_{t}x'}\outermodel_{x'x_{t-1}}\right)\nu_{x_0}.
\intertext{Note that reordering the sums in this expression is unproblematic, since all summands are non-negative. Applying \eqref{eq:helper_H}, we thus find}
  G_\photons(\tau) 
    &= \sum_{x\in\states^{T+1}}\left(\prod_{t=1}^T H_{x_tx_{t-1}}(\tau_t)\right)\nu_{x_0} \\
    &= (1, \dots, 1) \,H(\tau_T)\cdots H(\tau_1)\,\nu
\end{align*}
for the moment generating function. 
\end{proof}

\paragraph{Moments of the inner model.} 
Before we derive the expectation and covariance of $\photons$ by differentiating
equation \eqref{eq:generating_function}, we introduce three parameters
$\theta_1$, $\theta_2$, and $\theta_3$ that describe the photon emission
statistics $p_{00}$ and $p_{10}$ up to second order. The first parameter
$\theta_1$ describes the
expected number of photons emitted during the frame if the fluorophore was in
the bright state initially,
\begin{subequations}
\begin{equation}\label{eq:theta_1}
  \theta_1 
    = \EE\big[\photons_t \,\big|\, \state'_t = 0\big] 
    = \sum_{x\in\states} q_{x0} \,\EE\big[\photons_t\,\big|\,\state'_t = 0, \state_t = x],
\end{equation}
where we used that $q_{x0} = \innermodel_{x0} = \PP(\state_t
= x\,|\,\prestate_t = 0)$
by definition. The second parameter $\theta_2$ quantifies the contribution
to the expectation $\theta_1$ if the fluorophore not only starts the frame in the
bright state but also stays there,
\begin{equation}\label{eq:theta_2}
    \theta_2 = \frac{q_{00}\,\EE\big[\photons_t\,\big|\,\state'_t = \state_t = 0\big]}{\theta_1}.
\end{equation}
Finally, we capture the conditioned variance of $\photons_t$ given $\prestate_t
  = 0$ via the parameter $\theta_3$,
\begin{equation} \label{eq:theta_3}
\theta_3 = \frac{\Var\big[\photons_t \,\big|\, \state'_t = 0\big]}{\theta_1^2} - \frac{1}{\theta_1}.
\end{equation}
\label{eq:theta}%
\end{subequations}
This parameter can be viewed as the excess relative variance with respect to
a Poisson distribution: if $\photons_t \,|\, \state_t = 0$ was distributed
Poissonian, then $\theta_3 = 0$. 
A Poissonian statistics is often assumed as an approximation for the photon
emission of fluorophores, but corrections may be necessary for accurate results
\cite{aspelmeier2015}.

\begin{remark}{}{}
The parameters $\theta = (\theta_1, \theta_2, \theta_3)$ clearly depend on the
inner model choice and are usually related to the entries of the short-time
transition matrix $\innermodel$. For example, we show in Section \ref{sec:alexa}
that $\theta_2$ is fully determined by $q_{00}$ in the inner model that we use
for the fluorophore Alexa 647.
\end{remark}

\paragraph{Expectation.}
In order to derive analytical expressions for the expectation and covariance, we
will assume that the transition matrix $\totalmodel = \innermodel\outermodel$ is
diagonalizable and has eigenvalues $\lambda_0, \ldots, \lambda_r \in \CC$\label{eigenvalues}.
We argue that this assumption is no significant restriction, since the
stochastic matrices that are not diagonalizable form a null set in a reasonable
sense -- see Lemma \ref{lem:diag} in Appendix \ref{sec:diag} for details.
We thus write
\begin{equation}\label{eq:diag}
  M = V \Lambda V^{-1},
\end{equation} 
where $\Lambda = \mathrm{diag}(\lambda_0, \dots, \lambda_r)$ and where
$V$ is a matrix containing the eigenvectors of $M$ as columns. 
Due to the absorbing nature of the bleached state $x = r$, we can assume that
$\lambda_r = 1$ with eigenvector $(0, 0, \dots, 1)^\mathrm{T}$, which hence
constitutes the last column of $V$.

\begin{theorem}{}{singleexpectation}
  Assume that the Markov matrix $M$ for a single-fluorophore model in
  $\singlemodels$ is diagonalizable like in \eqref{eq:diag}.
  Then the expectation value $\mu_t$ of the number $Y_t$ of photons emitted by
  the fluorophore at time $t = 1, \ldots, T$ is 
  \begin{equation}\label{eq:expectation}
    \mu_t = \theta_1\sum_{x\in\states} \alpha_x \lambda_x^{t-1},
  \end{equation}
  where the coefficients $\alpha_x$ for $x\in\states$ are defined by
  \begin{equation}\label{eq:alpha}
    \alpha_x = V^{}_{0x} \frac{\lambda^{}_x}{q_{00}} 
      \sum_{z\in\states}V^{-1}_{xz} \,\nu_z.
  \end{equation}
\end{theorem}

\begin{proof}
Upon differentiating the moment generating function $G_\photons(\tau)$ in equation
\eqref{eq:generating_function} with respect to $\tau_t$, one obtains the
expectation value of $\photons_t$,
\begin{align}
  \mu_t = \EE[\photons_{t}]
&= \frac{\partial G_\photons(\tau)}{\partial \tau_t}\,\Big|_{\tau=0} \nonumber \\
&= (1,\dots,1) \,H(\tau_T)\cdots H(\tau_{t+1})H'(\tau_t)H(\tau_{t-1})\cdots H(\tau_1) \,\nu_0\,\big|_{\tau=0} \nonumber \\
&= (1,\dots,1)\,H'(0)M^{t-1} \,\nu_0 \label{eq:mean_derivation}.
\end{align}
The last line follows from $(1,\dots,1)\,H(0)=(1,\dots,1)$, which holds since
$H(0) = \innermodel \outermodel = \totalmodel$ is a probability matrix.
The derivative of $H$ is given by
\begin{equation*}
  H'(s) = (G'(s)\circ \innermodel)\outermodel,
\end{equation*}
where $G'$ is the derivative of the generating matrix $G$ defined in equation
\eqref{eq:great_g}.
Due to the particular form of $\outermodel$, $\innermodel$, and $G(s)$, see
equations \eqref{eq:modelmatrices} and \eqref{eq:great_g}, it follows with 
definition \eqref{eq:theta_1} of $\theta_1$ that 
\begin{align}
  (1,\dots,1)\,H'(0) &= \theta_1 (1, 0, \dots, 0)\,\outermodel \nonumber \\
  &= \frac{\theta_1}{q_{00}}(1, 0, \dots, 0)\,\totalmodel. \label{eq:aux_mean}
\end{align}
Here, we used that the respective first rows of $\outermodel$ and
$\totalmodel$ only differ by the factor $q_{00}$. Combining equations
\eqref{eq:mean_derivation} and \eqref{eq:aux_mean}, we can express the
expectation by
\begin{equation}\label{eq:mean_general}
  \mu_{t} = \frac{\theta_1}{q_{00}} \sum_{z\in\states} M^t_{0z} \,\nu_z.
\end{equation}
If we now use representation \eqref{eq:diag} of $M$, we find
\begin{align}
\mu_{t} &= \frac{\theta_1}{q_{00}} \sum_{z\in\states} 
     \big(V\Lambda^tV^{-1}\big)_{0z} \,\nu_z \nonumber \\
  &= \frac{\theta_1}{q_{00}} \sum_{x\in\states} \left(V^{}_{0x}
     \sum_{z\in\states}V^{-1}_{xz} \nu_z\right) \lambda^{t}_x \nonumber \\
  &= \theta_1 \sum_{x\in\states} \alpha_x \lambda^{t-1}_x, \label{eq:mean}
\end{align}
which proves the theorem.
\end{proof}

When deriving the covariance of $\photons$ later on, we will need the
expectation value $\mu^0_t$ of $Y_t$ on the condition that the fluorophore was
in the bright state at the beginning of the experiment. This corresponds to the
case $\nu = (1,0,\ldots,0)$. According to equation \eqref{eq:mean_general}, we
thus find
\begin{equation}\label{eq:mu0}
  \mu_{t}^0 = \EE[Y_t \,|\, X_0 = 0] = \frac{\theta_1}{q_{00}}\big(M^{t}\big)_{00}.
\end{equation}
Under assumption \eqref{eq:diag} of diagonalizability for $M$, Theorem
\ref{thm:singleexpectation} lets us write
\begin{equation*}
  \mu^0_t = \theta_1 \sum_{x\in\states} \alpha^0_x \lambda_x^{t-1},
\end{equation*}
where the coefficients $\alpha^0_x$ are given by equation \eqref{eq:alpha} with
$\nu = (1, 0, \dots, 0)$,
\begin{equation}\label{eq:alpha0}
  \alpha^0_x = \frac{\lambda_x}{q_{00}} V^{}_{0x}V^{-1}_{x0}.
\end{equation}
These auxiliary coefficients $\alpha^0_x$ can be related
to $\alpha_x$.  If $\nu_0 = 1$, then $\alpha^0_x = \alpha_x$ by definition. If
$\nu_0 < 1$, we can plug $\nu' = \frac{1}{1-\nu_0}(0, \nu_1, \dots, \nu_r)$ in equation
\eqref{eq:alpha} and define
\begin{equation}\label{eq:alpha1}
  \alpha^1_x = \frac{\lambda_x}{q_{00}} V_{0x}^{}\sum_{z\in\states}V^{-1}_{xz}\nu'_z,
\end{equation}
which allows us to decompose $\alpha_x$ as
\begin{equation*}
  \alpha_x = \nu_0\,\alpha^0_x + (1-\nu_0)\,\alpha^1_x.
\end{equation*}
This way of splitting up the model parameters has the advantage that
a simple set of constraints for $\alpha^0_x$ and $\alpha^1_x$ arises (see Lemma
\ref{lem:constraints} below). 

\paragraph{Spectral properties and parameter constraints.}
The eigenvalues $\lambda_x$ and coefficients $\alpha_x$ in Theorem
\ref{thm:singleexpectation} can in general be complex-valued.
When estimating these parameters numerically, however, it is beneficial to
assume real eigenvalues and eigenvectors of $M$. In Appendix \ref{sec:eig}, we
provide some criteria that guarantee $\lambda_x\in[0, 1]$ for $r \le 3$. In
summary, (i) real and (ii) positive eigenvalues are ensured if the diagonal
values of the transition matrix $M$ are (i) diverse and (ii) large enough.
Usually, both of these assumptions are physically reasonable: 
the diagonal values are diverse if the outer states of the fluorophore exhibit
diverse live times, and they are large if the outer states are on average stable
over several frames.
Under the restriction $\lambda_x\in[0, 1]$ on the spectrum of $M$, equation
\eqref{eq:expectation} states that the expected number $\mu_t$ of emitted
photons is the superposition of exponential decays with timescales determined by
$\lambda_x$.

Furthermore, note that the coefficients $\alpha_x$ are implicitly constrained by
their definition in equation \eqref{eq:alpha}. First, 
\begin{equation*}
  \alpha_r = 0
\end{equation*}
is enforced due to $V_{0r} = 0$. This is physically expected as fluorophores in
the bleached state do not emit photons. Secondly, summing over $x\in\states$ in
\eqref{eq:alpha} shows
\begin{equation}\label{eq:alpha_constraint_1}
  \sum_{x\in\states} \alpha_x 
    = \frac{1}{q_{00}} \sum_{z\in\states} M_{0z} \nu_z 
    = (1, q_{01}, \dots, q_{0(r-1)}, 0)\,\nu.
\end{equation}
Similarly, we find the relation
\begin{equation}\label{eq:alpha_constraint_2}
  \sum_{x\in\states} \frac{\alpha_x}{\lambda_x} = \frac{\nu_0}{q_{00}}
\end{equation}
by dividing equation \eqref{eq:alpha} by $\lambda_x$ and again summing over
$x\in\states$. Applying the last three equations to the coefficients
$\alpha^0_x$ and $\alpha^1_0$, defined in \eqref{eq:alpha0} and
\eqref{eq:alpha1}, yields a set of simple constraints.

\begin{lemma}{}{constraints}
  We have $\alpha_r^0 = \alpha_r^1 = 0$. Furthermore, it holds that
  \begin{subequations}
  \begin{align}
    \sum_{x\in\states} \alpha^0_x = 1 
      \qquad&\text{and}\qquad 
    \sum_{x\in\states}\frac{\alpha^0_x}{\lambda_x} = \frac{1}{q_{00}},
    \label{eq:alpha_constraint_general_1} \\[0.25cm]
    0 \le q_{00}\sum_{x\in\states}\alpha^1_x \le 1
      \qquad&\text{and}\qquad
    \sum_{x\in\states}\frac{\alpha^1_x}{\lambda_x} = 0.
    \label{eq:alpha_constraint_general_2}
  \end{align}
  \label{eq:alpha_constraints}%
  \end{subequations}
\end{lemma}
\begin{proof}
  The first statement holds due to $V_{0r} = 0$.
  The relations in \eqref{eq:alpha_constraint_general_1} follow from equation
  \eqref{eq:alpha_constraint_1} and \eqref{eq:alpha_constraint_2} for
  $\nu = (1, 0, \ldots, 0)$. Constraint \eqref{eq:alpha_constraint_general_2}
  follows similarly if we take into account that $\nu'$ is a probability vector
  with $\nu'_0 = 0$.
\end{proof}

\paragraph{Covariance.} We next look at the covariance matrix $\Sigma$ of the
process $\photons$, which can be obtained from the second derivatives of the
moment generating function $G_\photons$.

\begin{theorem}{}{singlecovariance}
  Under assumption \eqref{eq:diag}, the covariance matrix $\Sigma$ of the
  process $Y = (Y_t)_{t=1}^T$ is
  \begin{subequations}
  \begin{align}
    \Sigma_{tt}  &= 
      \big(\theta_1(\theta_3 + 1) + 1 - \mu_t\big)\,\mu_t,
    \label{eq:cov_diag} \\[0.25cm]
    \Sigma_{tt'} &=
      \left[ \left(\theta_2 - q_{00}
      \frac{1-\theta_2}{1-q_{00}}\right) \mu^0_{t-t'} + 
      \frac{1-\theta_2}{1-q_{00}}\mu^0_{t-t'+1} - \mu_t\right] \mu_{t'},
    \label{eq:cov_offdiag}
  \end{align}
  \label{eq:cov}
  \end{subequations}
  on the diagonal and off-diagonal with $t > t'$, respectively.
\end{theorem}

\begin{proof}
One can derive the entries of the covariance matrix for times $t, t' = 1,
\ldots, T$ by
\begin{equation}\label{eq:cov_derivation}
  \Sigma_{tt'} = \frac{\partial^2G_\photons(\tau)}{\partial\tau_t
    \partial\tau_{t'}}\,\Big|_{\tau=0} - \mu_t \, \mu_{t'}.
\end{equation}
We first address the diagonal with $t = t'$. In this case, we can
proceed similarly to equation \eqref{eq:mean_derivation} and find
\begin{equation}\label{eq:cov_diag_derivation}
  \Sigma_{tt} = (1, \dots, 1)\,H''(0)M^{t-1}\,\nu - \mu_t^2.
\end{equation}
Again, one can exploit the special forms of $\outermodel$, $\innermodel$,
and $G(\tau)$ to obtain
\begin{align*}
  (1,\dots,1)\,H''(0)
    &= \theta_1^2(\theta_3 + 1/\theta_1 + 1)\,(1, 0, \dots, 0)\,\outermodel \\
    &= \frac{\theta_1^2(\theta_3 + 1/\theta_1 + 1)}{q_{00}}\,(1, 0, \dots, 0)\,
    \totalmodel,
\end{align*}
where we used the relation between $\theta_1$, $\theta_3$, and the second
moments $G''_{00}(0)$ and $G''_{10}(0)$. Consulting equations
\eqref{eq:mean_derivation} and \eqref{eq:aux_mean} now reveals
\begin{equation*}
 (1, \dots, 1)\,H''(0)M^{t-1}\,\nu = \theta_1(\theta_3 + 1/\theta_1 + 1)\,\mu_t.
\end{equation*}
Plugging this expression in equation \eqref{eq:cov_diag_derivation} shows result
\eqref{eq:cov_diag}.

We next consider the off-diagonal entries with $t > t'$. Applying equation
\eqref{eq:cov_derivation} yields
\begin{equation*}
  \Sigma_{tt'} = (1,\dots, 1)\, H'(0)M^{t-t'-1}H'(0)M^{t'-1}\nu - \mu_t \, \mu_{t'}.
\end{equation*}
Since $H'(0) = \big(G'(0) \circ \innermodel\big)\outermodel$, and since only the
first column of $G'(0)$ is unequal to $0$, we find that
\begin{align*}
  (1,\dots, 1)\, H'(0)M^{t-t'-1}H'(0) 
  &= \helper\,\theta_1 \,(1, 0, \dots, 0)\,\outermodel \\
  &=  \helper\frac{\theta_1}{q_{00}} (1, 0, \dots, 0)\,\totalmodel,
\end{align*}
where $\helper$ is given by
\begin{align}
  \helper 
  &= \frac{1}{\theta_1}(1,\dots, 1)\,H'(0)M^{t-t'-1}\big(G'(0)\circ \innermodel\big)\, (1, 0, \dots, 0)^\mathrm{T} \nonumber \\
  &= \frac{1}{q_{00}}(1,0,\dots, 0)\,M^{t-t'}\big(G'(0)\circ \innermodel\big)\, (1, 0, \dots, 0)^\mathrm{T}. \label{eq:helper}
\end{align}
We employed equation \eqref{eq:aux_mean} for the second equality.
Remarkably, we can now use the same reasoning as for the expectation and
the diagonal entries before, and we find that
\begin{equation}\label{eq:sigma_helper}
  \Sigma_{tt'} = \helper\,\mu_{t'} - \mu_t\, \mu_{t'}.
\end{equation}
This reduces the problem to resolving $\helper$. We begin by looking at the
first column of $G'(0) \circ \innermodel$, which is given by
\begin{equation}\label{eq:split}
  \big(G'(0)\circ \innermodel\big)\, 
  \begin{pmatrix}
    1 \\
    0 \\
    \vdots \\
    0
  \end{pmatrix}
  =
  \begin{pmatrix}
    G_{00}'(0)\,q_{00} \\
    G_{10}'(0)\,q_{10} \\
    \vdots \\
    G_{10}'(0)\,q_{r0} \\
  \end{pmatrix}
  =
  \theta_1\theta_2\,
  \begin{pmatrix}
    1 \\
    0 \\
    \vdots \\
    0 \\
  \end{pmatrix}
  +
  \theta_1\frac{1 - \theta_2}{1 - q_{00}}
  \begin{pmatrix}
    0 \\
    q_{10} \\
    \vdots \\
    q_{r0} \\
  \end{pmatrix}.
\end{equation}
Here, we used that 
$G_{xx'}'(0) = \EE\big[\photons_t\,\big|\,\state_t = x, \prestate_t = x'\big]$
and applied definitions (\ref{eq:theta}a-b) of $\theta_1$ and
$\theta_2$. 
The assumption of a common exit state in the fast model (see Remark
\ref{rem:commonexit}), which ensures that $G_{10} = G_{x0}$ for all $x \ge 1$,
is crucial for this step. Equation \eqref{eq:split} decomposes $\helper$ into
two parts, and we accordingly write
\begin{equation}\label{eq:decompose-helper}
  \helper = \theta_2 \,\helper_1 + 
  \frac{1 - \theta_2}{1 - q_{00}}\,\helper_2.
\end{equation}
We address $\helper_1$ first. By inserting the first term of \eqref{eq:split} in
\eqref{eq:helper}, we find that 
\begin{equation*}
  \helper_1 
    = \frac{\theta_1}{q_{00}}(1,0,\dots,0)\,M^{t-t'}(1,0,\dots,0)^\mathrm{T} 
    = \frac{\theta_1}{q_{00}}\big(M^{t-t'}\big)_{00}
    = \mu^0_{t-t'},
\end{equation*}
where definition \eqref{eq:mu0} of $\mu^0_t$ was applied.
In order to express $\helper_2$, we first note that the respective first columns
of the two matrices $\innermodel$ and $\totalmodel$ are equal, namely $(q_{00},
\dots, q_{r0})^\mathrm{T}$. Thus, we can express $q_{x0}$ in terms of the
diagonal decomposition \eqref{eq:diag} of $M$,
\begin{equation*}
  q_{x0} = \totalmodel_{x0} = \big(V \Lambda V^{-1}\big)_{x0}.
\end{equation*}
We then harness the auxiliary calculation
\begin{align*}
  \sum_{z=1}^r V_{xz}^{-1} q_{z0} = \sum_{z=1}^r V^{-1}_{xz} \big(V\Lambda V^{-1})_{z0} = \lambda_x V^{-1}_{x0} - q_{00}\,V^{-1}_{x0}, 
\end{align*}
which can be verified by straightforward computation, and arrive at
\begin{align*}
  \helper_2 
  &= \frac{\theta_1}{q_{00}}(1, 0, \dots, 0)\,M^{t-t'} ( 0, q_{10}, \dots, q_{r0} )^\mathrm{T} \\
  &= \frac{\theta_1}{q_{00}}(1, 0, \dots, 0)\,V\Lambda^{t-t'}V^{-1} ( 0, q_{10}, \dots, q_{r0} )^\mathrm{T} \\
  &= \frac{\theta_1}{q_{00}}(1, 0, \dots, 0)\,\Big(M^{t-t'+1} - q_{00}\, M^{t-t'}\Big) (1, 0, \dots, 0)^\mathrm{T}.
\end{align*}
Making use of the definition \eqref{eq:mu0} of $\mu_t^0$, we conclude
$\helper_2 = \mu_{t-t'+1}^0 - q_{00}\,\mu_{t-t'}^0$.
Together with $\helper_1 = \mu^0_{t-t'}$, the decomposition
\eqref{eq:decompose-helper} of $\helper$ can now be resolved to read
\begin{equation}\label{eq:helper_resolved}
  \helper = \left( \theta_2 - q_{00} \frac{1-\theta_2}{1-q_{00}}\right) \mu_{t-t'}^0 +
  \frac{1 - \theta_2}{1 - q_{00}} \,\mu_{t-t'+1}^0,
\end{equation}
which completes expression \eqref{eq:sigma_helper} and proves
\eqref{eq:cov_offdiag} for the off-diagonal elements of the covariance matrix.
\end{proof}

\begin{figure}[t]
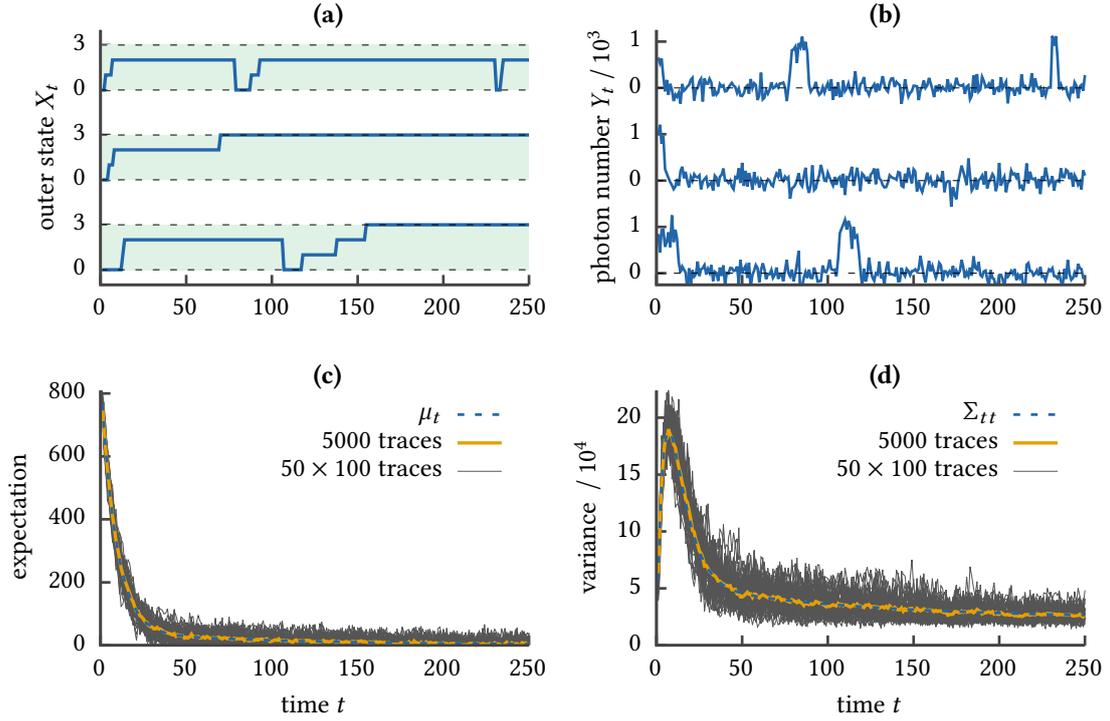

  \centering\small
  \input{traces.tex}

  \vspace{0.3cm}

  \input{meanvar.tex}
  \caption{Simulation results of the single-fluorophore HTMM. \textbf{(a)} shows
    three exemplary paths of the fluorophore in the outer state space $\states
    = \{0,\ldots,r\}$ for $r = 3$.
    One can see that two of the three fluorophores have already bleached in the
    first 250 time steps.
    \textbf{(b)} shows three fluorescence traces $y = (y_t)_{t=1}^{250}$
    corresponding to the paths in (a).
    To obtain more realistic traces, Gaussian white noise with mean $0$
    and standard deviation $\theta_1/5$ was added to each observation $y_t$.
    \textbf{(c-d)} 
    show the theoretical expectation $\mu_t$ and variance $\Sigma_{tt}$ of $Y_t$
    compared to their empirical estimates for $5000$ and $100$ simulated traces.
    In all simulations, we use the inner model described in Section
    \ref{sec:alexa}. The parameters for the inner and outer model are chosen
    such that the resulting traces roughly resemble the experimental data in
    \cite{laitenberger_2018}.
  }
  \label{fig:simulations}
\end{figure}

\begin{figure}
  \centering\small

  \vspace{-0.5cm}

  \input{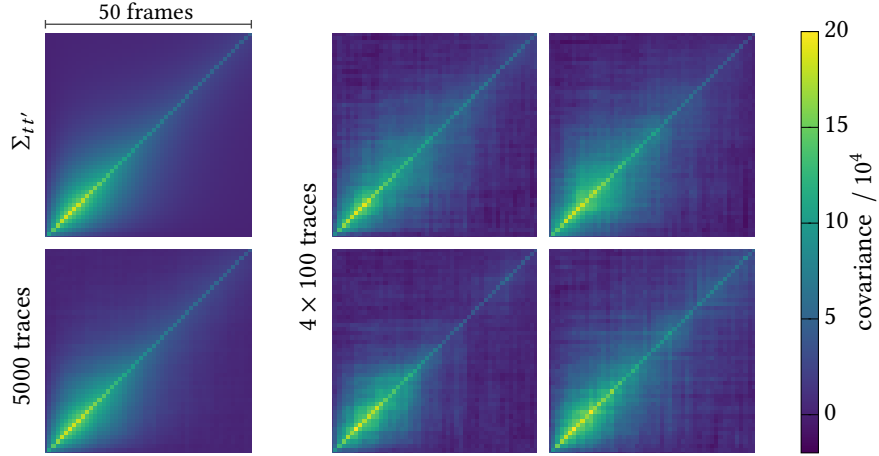}
  \caption{Real and empirical covariance of the HTMM. Shown are the covariances
    $\Sigma_{tt'}$ for $t,t' = 1, \ldots, 50$ (upper left image) as well as their
    empirical counterparts for 5000 traces (lower left image) and 100 traces
    (four images on the right). The same simulated traces as in Figure
    \ref{fig:simulations}c-d are used.
  }
  \label{fig:covariances}
\end{figure}

\begin{remark}{}{expcov}
  The expectation and covariance in equations \eqref{eq:expectation} and
  \eqref{eq:cov} depend on the eigenvalues and eigenvectors of the matrix
  $\totalmodel = \innermodel\outermodel$, but not on the single transition
  probabilities in $\innermodel$ and $\outermodel$ directly. From joint
  knowledge of $\mu$ and $\Sigma$, the parameters $\nu_0, q_{00},
  \alpha^0_x, \alpha^1_x, \lambda_x, \theta_1, \theta_2$, and $\theta_3$ are
  identifiable.
  Not all of them, however, are independent (see Lemma \ref{lem:constraints}),
  and knowing these parameters is in general not sufficient to
  reconstruct the matrices $\innermodel$ and $\outermodel$. Plots of $\mu$ and
  $\Sigma$ as well as simulation results for the processes $X$ and $Y$ are
  depicted in Figure \ref{fig:simulations} and \ref{fig:covariances}.
\end{remark}

\begin{remark}{variance \enquote{dip}}{dip}
  Figure \ref{fig:simulations}d illustrates a characteristic property of the
  variance $\Sigma_{tt}$ in dependence of the frame number $t$. Initially, the
  variance increases for some frames before it subsequently relaxes towards the
  background noise exponentially. This salient \enquote{dip} in the variance
  curve is also observed in experimental data \cite{laitenberger_2018} for large
  values of $\nu_0$, i.e., if most fluorophores are bright at the beginning of
  the experiment.
  It is caused by bright fluorophores getting dark during the first few frames,
  such that the observable distribution of photons $Y_t$ is effectively an
  additive model composed of two parts: dark fluorophores with $Y_t = 0$
  and bright ones with $Y_t$ around $\theta_1$. This split in the distribution
  of $Y_t$ temporarily causes a high variance until the number of dark
  fluorophores eventually dominates in the long run.
\end{remark}

\paragraph{Multiple fluorophores.} The signal we observe in experiments is based
on the fluorescent activity of an unknown number $m$\label{firstm} of fluorophores.
As we will typically not be able to distinguish between the contributions
resulting from different fluorophores, we can only rely on the total number
$\photons^{(m)}_t$ of photons emitted in frame $t$. It is given by the sum of
$m$ single-fluorophore processes $Y^k$ modeled via $\singlemodels$,
\begin{equation}
  Y^{(m)}_t = \sum_{k = 1}^{m} Y^k_t.
\end{equation} 
We make the assumption that the contributions $Y^k$ are independent and
identically distributed (see Remark \ref{rem:iid_assumptions}). 
Even though these assumptions are approximations -- conditions like the
biochemical properties of the fluorophore's neighbourhood or its spatial
orientation have a certain impact -- 
they appear to lead to a decent
description for the multi-fluorophore dynamics in practical situations
\cite{laitenberger_2018}. 
The set of all multi-fluorophore models that obey the i.i.d.\ assumptions is
henceforth denoted by $\multimodels$.

The expectation and covariance of $Y^{(m)}$ as sum of $m$ i.i.d.\ random
processes simply acquire the prefactor $m$ with respect to the
single-fluorophore expressions.
Note that we will use the same symbols $\mu$, $\mu^0$, and $\Sigma$ as for the
single-fluorophore process, see \eqref{eq:expectation}, \eqref{eq:mu0}, and
\eqref{eq:cov}, to denote the respective generalizations to $m \ge 1$
fluorophores.
\begin{theorem}{}{multimodel}
  The expectation $\mu$ and covariance $\Sigma$ of a multi-fluorophore process
  $Y^{(m)}$ in $\multimodels$ are
  \begin{subequations}
  \begin{equation}
    \mu_t = m\,\theta_1 \sum_{x\in\states} \alpha_x\lambda_x^{t-1} =  m\,\theta_1 \sum_{x=0}^{r-1} 
    \big(\nu_0\,\alpha_x^0 + (1 - \nu_0)\,\alpha^1_x\big)\lambda_x^{t-1},
    \label{eq:multi_mean}
  \end{equation}
  and
  \begin{align}
    \Sigma_{tt} &= 
    \frac{1}{m} \big(m\,\theta_1(\theta_3 + 1) + m - \mu_t\big) \mu_t,
    \label{eq:multi_cov_diag} \\[0.25cm]
    \Sigma_{tt'} &=
    \frac{1}{m} \left[ \left(\theta_2 - q_{00} \frac{1-\theta_2}{1-q_{00}}\right)\mu^0_{t-t'} + \frac{1-\theta_2}{1-q_{00}}\mu^0_{t-t'+1} - \mu_t\right] \mu_{t'},
    \label{eq:multi_cov_offdiag}
  \end{align}
  where
  \begin{equation}
    \mu_t^0 = m\,\theta_1 \sum_{x\in\states} \alpha^0_x\lambda_x^{t-1}, \hspace{5.50cm}
    \label{eq:multi_mean0}
  \end{equation}
  for $t,t' = 1, \ldots, T$ with $t > t'$. The coefficients $\alpha_x$, $\alpha^0_x$, 
  and $\alpha^1_x$ are given by equations \eqref{eq:alpha}, \eqref{eq:alpha0}, and
  \eqref{eq:alpha1}, respectively.
  \label{eq:multi}
  \end{subequations}
\end{theorem}

\paragraph{Parameterization.} 
At this point, it is instructive to think about the parameterization of the
multi-fluorophore model class $\multimodels$ (see also Table \ref{tab:notation}
on page \pageref{tab:notation}). The full model for the photon
emission process $Y^{(m)}_t$ depends on the fluorophore number $m$, on all
transition probabilities $Q = (q_{xz})$, on the initial state $\nu$,
and on an unspecified number of parameters that come with a concrete choice of
the inner model. If we only want to describe the first two moments $\mu$ and
$\Sigma$, however, several parameters become hidden and are not
required to be estimated for our purposes.

According to equations (\ref{eq:multi}a-d), we only need the fluorophore number
$m$, the fraction of initially bright fluorophores $\nu_0$, the probability
$q_{00}$ for a fluorophore to stay bright during the exposure, and the
eigenvalues $\lambda = (\lambda_x)_{x\in\states}$ as well as
coefficients $\alpha^0 = (\alpha^0_x)_{x\in\states}$ and $\alpha^1
= (\alpha^1_x)_{x\in\states}$ for the multi-exponential decay in order to
calculate the expectation and covariance. The parameters $m$ and $\nu_0$
contribute one degree of freedom each, while $q_{00}$ is fixed by $\lambda$ and
$\alpha^0$ due to equation \eqref{eq:alpha_constraint_general_1}. To specify
$\lambda$, assuming it is real, we need $r$ free components since $\lambda_r = 1$
is determined through the bleached state. Similarly, $\alpha^0_r = \alpha^1_r
= 0$. Due to the constraints \eqref{eq:alpha_constraints} in Lemma
\ref{lem:constraints}, the parameters $\alpha^0$ and $\alpha^1$ hence contribute
$r-1$ free components each.  This makes a total number of $3r$ independent
parameters, compared to $r^2 + (r-1)$ degrees of freedom needed to specify all
components of $Q$ and $\nu$. 

The three parameters $\theta = (\theta_1, \theta_2, \theta_3)$
are sufficient to specify the effect of the inner model in the second-order
description. Still, specific knowledge of the inner model is necessary,
since the relation of $\theta$ to other parameters is unclear otherwise.
For example, an inner model with a Poissonian photon statistics 
$\photons_t\,|\,\prestate_t = 0$ enforces $\theta_3 = 0$, which evidently
reduces the number of free parameters. Similarly, $\theta_2$ is not a free
parameter for the inner model that we employ to describe the fluorophore Alexa
647 in Section \ref{sec:alexa}; it is completely determined by $q_{00}$.

A setting we want to emphasize is the one where $\nu_0 = 1$, i.e., where each
fluorophore is bright at the beginning of the experiment. This can be enforced
by the experimental setup, like in the super-resolution scheme applied in
\cite{laitenberger_2018}. Then, the $r - 1$ parameters that are needed to
specify the coefficients $\alpha^1$ drop out of the formulae for $\mu$ and
$\Sigma$, which makes this choice particularly beneficial.

\section{Image Acquisition}
\label{sec:thinning}

In the previous section, we introduced an elaborate statistical model
$\multimodels$ for the number of photons that are generated by a set of $m$
fluorophores during a series of exposures in super-resolution microscopy.
We next look at the image acquisition procedure and discuss the relation
between the photon emission process $\photons^{(m)}$ and the final time trace
$\ccdvals^{(m)}$ captured by the CCD camera.
Fortunately, most processing steps subsequent to the emission of photons -- like
thinning in the microscope or amplification through the CCD camera -- can be
included into the model by modifying the photon statistics $p_{00}$ and
$p_{10}$, see \eqref{eq:great_p}.
Consequently, merely the parameters $\theta = (\theta_1, \theta_2, \theta_3)$
will be affected in our second-order description, and equations
(\ref{eq:multi_mean}--\ref{eq:multi_mean0}) for the expectation $\mu$ and
the covariance $\Sigma$ will remain intact: we just need to substitute
$\theta$ by suitable transformed parameters $\theta'$.\footnote{
  This is not entirely accurate. Camera noise contributions that do not
  depend on the fluorophore and its state of activity, which we called
  $\epsilon_t$ in equation \eqref{eq:camera} of Section \ref{sec:modelingsteps},
  cannot be modeled that way and have to be considered separately.  Their
  inclusion in the inner model would require a signal $Y_t > 0$ even for
  fluorophores in dark states $X'_t > 0$, which we explicitly prohibited during our
  derivations of $\mu$ and $\Sigma$ in the previous section.
}

It might thus seem superfluous to explicitly model
any further steps in the microscope and camera, since we will typically
estimate $\theta$ from the data anyway. However, there are several reasons
why it is important to understand how the original parameter $\theta$ is
transformed to $\theta'$.
First, these transformations could alter the constraints placed on
parameters by the inner model (like $\theta_3 = 0$ if $\photons_t
\,|\, \prestate_t = 0$ is Poissonian) by possibly introducing new parameters
(such that $\theta'_3$ could be a free parameter again, e.g., due to an
unknown variance of the amplification for the specific camera model).
Second, the relation between $\theta$ and $\theta'$ could be interesting in
its own right, because $\theta$ contains immediate information about
the actual physics of the fluorophore, while $\theta'$ merges this
information with further details of the experimental setup. This additional
degree of insight could also be helpful for Bayesian inference approaches,
where prior knowledge about the parameters is taken into account.

\paragraph{Thinning.}
In Section \ref{sec:modelingsteps}, we mentioned that only a certain fraction
of emitted photons hits the detector interface and is registered at
some CCD pixel. Many photons will fail to reach the optical pathway or will be
absorbed by the equipment (like lenses or mirrors). 
The probability that an emitted photon triggers a photoelectron in
a specified region of interest $R$ of the camera was denoted by
$p_\mathrm{d}$, such that 
\begin{equation}\label{eq:thinned-general}
\thinned_t \sim \mathrm{Bin}(\photons_t, p_\mathrm{d})
\end{equation}
models the thinned photon number for a single fluorophore. 
The parameter transformation from $\theta$ to $\theta'$ that accompanies
this thinning process can be established by plugging $\thinned_t$ in
equations (\ref{eq:theta}a-c) defining the inner parameters.
\begin{lemma}{}{thinning}
  The moment parameters $\theta'$ of the thinned photon counts $\thinned_t$
  are given by
  \begin{equation}\label{eq:theta_3_thinning}
    \theta_1' = p_\mathrm{d} \theta_1, 
    \quad \theta_2' = \theta_2, 
    \quad \theta_3' = \theta_3.
  \end{equation}
  Therefore, only $\theta_1$ is affected by binomial thinning while
  $\theta_2$ and $\theta_3$ are left unaffected.
\end{lemma}
\begin{proof}
Employing definitions (\ref{eq:theta}a-c) yields 
\begin{align*}
  \theta_1' 
    &= \EE[\thinned_t\,|\, \prestate_t = 0] 
    = \EE\big[\EE[\thinned_t\,|\photons_t]\,\big|\,\prestate_t = 0\big] \\
    &= p_\mathrm{d}\, \EE[\photons_t\,|\,\prestate_t = 0] = p_\mathrm{d}\,\theta_1,
\intertext{as well as}
  \theta_2' &= q_{00}\,\frac{\EE[\thinned_t\,|\,\prestate = \state = 0]}{\theta_1'} \\
  &= q_{00}\,\frac{p_\mathrm{d}\,\EE[\photons_t\,|\,\prestate = \state = 0]}{p_\mathrm{d}\,\theta_1} = \theta_2,
\intertext{and, if we use the law of total variance,}
  \theta_3' 
    &= \frac{\Var[\thinned_t\,|\,\prestate_t = 0]}{(\theta'_1)^2} - \frac{1}{\theta'_1} \\
    &= \frac{\EE\big[\Var[\thinned_t\,|\,\photons_t] \,\big|\,\prestate_t = 0\big] + \Var\big[\EE[\thinned_t\,|\,\photons_t] \,\big|\,\prestate_t = 0\big] }{(\theta_1')^2} - \frac{1}{\theta_1'} \\
    &= \frac{p_\mathrm{d} (1 - p_\mathrm{d})\,\theta_1 + p_\mathrm{d}^2\Var[\photons_t\,|\,\prestate_t = 0]}{p_\mathrm{d}^2\,\theta_1^2} - \frac{1}{p_\mathrm{d}\,\theta_1} \\
    &= \frac{\Var[\photons_t\,|\,\prestate_t = 0]}{\theta_1^2} - \frac{1}{\theta_1} = \theta_3.
\end{align*}
\end{proof}

\paragraph{Signal amplification.}
When photoelectrons are read out in an EMCCD camera, an electron amplifying
system that consists of a cascade of electron multipliers (EM) is triggered.  Each
stage of this cascade has a certain probability of generating extra electrons,
and the succession of many stages results in a stochastic signal amplification
of the incident photons. This introduces additional noise, which we consider in
the following. 

We denote the distribution that results from the signal amplification of
a single photo electron by $\mathcal{D}$, like in Section
\ref{sec:modelingsteps}. For convenience, we use the symbol $\thinned$ to
denote a random variable with the stationary distribution of
$\thinned_t\,|\,\prestate_t=0$, i.e., we condition our considerations on
bright fluorophores.
Then, the number $\electrons$ of electrons generated by $\thinned$ detected
photons is
\begin{equation*}
  \electrons = \sum_{k=1}^{\thinned} U_k,
\end{equation*}
where $U_k \sim \mathcal{D}$. By the law of total variation we obtain
\begin{equation}
  \frac{\Var[\electrons]}{\EE[\electrons]^2} = \frac{\Var[U_1]}{\EE[U_1]^2} \frac{1}{\EE[\thinned]} + \frac{\Var[\thinned]}{\EE[\thinned]^2}.
\label{eq:excess_noise}
\end{equation}
If $\thinned$ was Poisson distributed with
parameter $\lambda > 0$, it would follow that 
\begin{equation*}
  \frac{\Var[\electrons]}{\EE[\electrons]^2} = \left(\frac{\Var[U_1]}{\EE[U_1]^2}  +  1\right)\frac{1}{\lambda}.
\end{equation*}
For this reason, the term 
\begin{equation}\label{eq:excessnoise}
  f^2 = \frac{\Var[U_1]}{\EE[U_1]^2}  +  1
\end{equation}
is called the ``excess noise factor''. For the amplification models
considered in \cite{robbins_noise_2003,hirsch_stochastic_2013}, we have
$1\le f^2 \le 2$.  The factor $f^2$ is usually known for a given camera. 

For each frame, the camera accumulates the photo electrons over a certain length
of time, the exposure time, before multiplying them. Afterwards, the accumulated
and multiplied electrons $E$ pass through the A/D converter,
which introduces a factor $c>0$ between the actual mean
number of amplified electrons and the output signal. The (ideal) output signal
$\ccdvals$ produced by a bright fluorophore is thus given by
\begin{equation}
  \ccdvals = c\electrons.
  \label{eq:camera_output}
\end{equation}
Similar to the case of thinning, this transformation of the photon
statistics corresponds to a transformation of the inner parameters of the
model from $\theta$ to $\theta'$. This time, we find the transformation
rules
\begin{equation} \label{eq:theta_3_camera}
  \theta_1' = a\,\theta_1,\qquad 
  \theta_2' = \theta_2,\qquad
  \theta_3' = \theta_3 + \frac{f^2 - 1}{\theta_1},
\end{equation}
where 
\begin{equation}\label{eq:ampfactor}
  a = c\, \frac{\EE[\electrons]}{\EE[\thinned]}
\end{equation}
is the overall amplification factor that translates from detected photons
$\thinned$ to the CCD output $\ccdvals$. Result \eqref{eq:theta_3_camera}
follows from utilizing relation \eqref{eq:excess_noise} in expression
\eqref{eq:theta_3}.

\paragraph{Offset and background.} 
Equation \eqref{eq:camera_output} is an idealization of the true camera output.
It neglects background photons in the setup as well as additional noise from the
analog circuits and the A/D converter. Furthermore, a positive offset is usually
applied to the pixel values in order to avoid noise induced fluctuations into
the negative domain. In contrast to our previous considerations, all of these
effects cannot be integrated into the parameters $\theta$ because they do not
exclusively affect the photon statistics $p_{00}$ and $p_{10}$. Instead, they
are independent of the state $\prestate_t$ of the fluorophore. The true
multi-fluorophore output signal observed in the region $R$ in frame $t$ is given
by
\begin{equation}\label{eq:full_model}
  \ccdvals_t^{(m)} = c\electrons_t + o + \epsilon_t,
\end{equation}
where $\electrons_t$ is the respective amplified number of electrons in
frame $t$, $o$ is an offset value, and $\epsilon_t$ is a centered random
variable that subsumes all additional noise sources and is considered to be
independent of $\electrons_t$.
The background noise $\epsilon_t$, whose standard deviation we denote by
$\sigma_t$\label{bgnoise},
can depend on time since the camera electronics may adapt during the experiment.
Together with the offset $o$, $\sigma_t$ can be estimated from the image series
directly.

We want to remark that the parameter $a$ in \eqref{eq:ampfactor} can also be
estimated directly \cite{hirsch_stochastic_2013}.  This can be done by
illuminating the camera with a temporally constant but spatially inhomogeneous
light intensity, which leads to Poisson statistics in each pixel with an
inhomogeneous parameter. One can then estimate the mean and variance of the
camera outputs $\ccdvals$ at each pixel from a time series of such images, and
determine $af^2$ as the slope in a plot of $\Var[\ccdvals]$ against
$\EE[\ccdvals]$ with known $f^2$, since
\begin{equation*}
  \Var[\ccdvals] = c^2\, \Var[E] + \Var[\epsilon] = c^2 f^2 \frac{\EE[\electrons]^2}{\EE [\thinned]} + \Var[\epsilon]
  = af^2\, \EE[\ccdvals] + \mathrm{const}.
\end{equation*}

\paragraph{Remarks and full model.} 
The preceding considerations show that the modeling steps of thinning and signal
amplification transform the inner model parameters $\theta$. Other
contributions that stem from the experimental setup, like noise due to
background photons, cannot be included in the description by merely modifying
$\theta$ and require an approach along the lines of model $\ccdvals^{(m)}$ in
equation \eqref{eq:full_model}. 
Some of the parameters that emerge are known (like the excess noise $f^2$) or
can be estimated independently from the fluorophore model (like the
amplification factor $a$, the offset $o$, and the variance $\sigma_t^2$ of the
background noise $\epsilon_t$).
The detection probability $p_\mathrm{d}$, however, cannot be
separated statistically from the expected number of photons $\theta_1$ during
inference, which is why we will drop $p_\mathrm{d}$ from the final model
formulation, effectively working with $p_\mathrm{d}\, \theta$ when we
write $\theta$ in the following. We also assume preprocessed image data, where
the offset $o$ has been subtracted and where the signal was divided by the total
amplification $a$, i.e., we actually consider normalized data that is modeled by
\begin{equation}\label{eq:normalized}
  \frac{\ccdvals^{(m)}_t - o}{a}.
\end{equation}
Our full second order description for fluorescence time traces is then captured
by the following theorem, the notation of which is summarized by Table
\ref{tab:notation} on page \pageref{tab:notation}.

\begin{theorem}{}{full}
The mean $\mu$ and the covariance $\Sigma$ of the normalized process
\eqref{eq:normalized}, which models the fluorophore activity observed by
a camera, are given by
\begin{subequations}
\begin{align}
  \mu_t &=  m\,\theta_1 \sum_{x=0}^{r-1} 
  \big(\nu_0\,\alpha_x^0 + (1 - \nu_0)\,\alpha^1_x\big)\lambda_x^{t-1},
  \label{eq:multi_mean_full} \\
  \Sigma_{tt} &= 
  \frac{1}{m} \big(m\,\theta_1(\theta_3 + 1) + mf^2 - \mu_t\big) \mu_t + \sigma_t^2/a^2,
  \label{eq:multi_cov_diag_full} \\[0.1cm]
  \Sigma_{tt'} &=
  \frac{1}{m} \left[ \left(\theta_2 - q_{00} \frac{1-\theta_2}{1-q_{00}}\right)
    \mu^0_{t-t'} + \frac{1-\theta_2}{1-q_{00}}\mu^0_{t-t'+1} - \mu_t\right] \mu_{t'},
  \label{eq:multi_cov_offdiag_full}
  \shortintertext{where}
  \mu_t^0 &= m\,\theta_1 \sum_{x=0}^{r-1} \alpha^0_x\lambda_x^{t-1},
  \label{eq:multi_mean0_full}
\end{align}
\label{eq:final-moments}%
\end{subequations}%
for $t,t' = 1, \ldots, T$ with $t > t'$.
\end{theorem}

\begin{proof}
  The respective expressions follow from combining equations (\ref{eq:multi}a-d)
  for the multi-fluorophore model $Y^{(m)}_t$ with result
  \eqref{eq:theta_3_camera} and definition \eqref{eq:full_model} established in
  this section. 
\end{proof}

\begin{table}
  \small\centering
  \begin{tabular}{@{}c@{~~~}ll|l|p{0.26\textwidth}}
    & \textbf{Symbol} & \multicolumn{1}{c}{\textbf{Meaning}} & \multicolumn{1}{c}{\textbf{Reference}} & \multicolumn{1}{c}{\textbf{Comment}} \\
    \toprule \\[-0.2cm]
    \multicolumn{5}{l}{\!\!\!\color{darkgray}\textit{stochastic fluorophore dynamics}} \\[0.2cm]
    & $\state_t$      & state after $t$-th exposure & p.\ \pageref{state} & \\
    & $\prestate_t$   & state before $t$-th exposure & p.\ \pageref{prestate} & \\[0.2cm]
    & $\photons_t$    & photons emitted in $t$-th exposure & p.\ \pageref{photons}& \\
    & $\ccdvals_t$    & camera output values in frame $t$ & eq.\ \eqref{eq:full_model}, p.\ \pageref{eq:full_model} & \\
    \\[0.1cm]
    \multicolumn{5}{l}{\!\!\!\color{darkgray}\textit{general model specification}} \\[0.2cm]
    & $\outermodel,\innermodel$   & long/short-time transition matrix & eq.\ \eqref{eq:deftransition}, p.\ \pageref{eq:deftransition} & 
      constrained in eq.\ \eqref{eq:modelmatrices} \\
    & $\nu = (\nu_x)$ & dist.\ of the initial state $\state_0$ & p.\ \pageref{initdist} & \\[0.2cm]
    & $q_{xz}$  & constrained transition probabilities & eq.\ \eqref{eq:modelmatrices}, p.\ \pageref{eq:modelmatrices} & entries of $\outermodel$ and $\innermodel$ \\
    & $p_{xz}$  & photon statistics & eq.\ \eqref{eq:photonstatistics}, p.\ \pageref{eq:photonstatistics}& 
      dist.\ of $\photons_t \,|\, \state_t = x, \prestate_t = z$, constrained in eq.\ \eqref{eq:great_p} \\
    \color{darkgray}\scriptsize$(*)$ & $m$ & number of i.i.d.\ fluorophores & p.\ \pageref{firstm}& central quantity of interest \\
    \\[0.1cm]
    \multicolumn{4}{l}{\!\!\!\color{darkgray}\textit{second order specification}} \\[0.2cm]
    \color{darkgray}\scriptsize$(*)$ & $\theta_1,\theta_2,\theta_3$ & inner model parameters & eq.\ \eqref{eq:theta}, p.\ \pageref{eq:theta} & 
    describe first two moments of $\photons_t \,|\, X_t \!\!=\!\! 0$, can be constrained (e.g., for Alexa 647) \\
    \color{darkgray}\scriptsize$(*)$ & $\nu_0$ & fraction of bright molecules at $t = 0$ & p.\ \pageref{initdist} & $\alpha^1$ drops out if $\nu_0 = 1$ \\
    \color{darkgray}\scriptsize$(*)$ & $q_{00}$ & prob.\ to stay bright in one exposure & eq.\ \eqref{eq:modelmatrices}, p.\ \pageref{eq:modelmatrices} & 
      usually connected to inner parameters $\theta$ \\
    \color{darkgray}\scriptsize$(*)$ & $\lambda = (\lambda_x)$ & eigenvalues of $\totalmodel = \innermodel\outermodel$ & p.\ \pageref{eigenvalues} &
      $\lambda_x\in[0,1]$ under suitable conditions, see appendix \ref{sec:eig} \\
                                      & $\alpha = (\alpha_x)$, & multi-exponential sum coefficients & eq.\ \eqref{eq:alpha}, p.\ \pageref{eq:alpha} & \\
     \color{darkgray}\scriptsize$(*)$ & $\alpha^0 = (\alpha^0_x)$, & & eq.\ \eqref{eq:alpha0}, p.\ \pageref{eq:alpha0} & \\
     \color{darkgray}\scriptsize$(*)$ & $\alpha^1 = (\alpha^1_x)$  & & eq.\ \eqref{eq:alpha1}, p.\ \pageref{eq:alpha1} & \\
    \\[0.1cm]
    \multicolumn{4}{l}{\!\!\!\color{darkgray}\textit{camera and background}} \\[0.2cm]
    & $f^2$ & excess relative variance of camera & eq.\ \eqref{eq:excessnoise}, p.\ \pageref{eq:excessnoise} & \\
    & $\sigma_t^2$ & background noise in frame $t$ & p.\ \pageref{bgnoise} & variance of $\epsilon_t$ in eq.\ \eqref{eq:full_model} \\
    & $a$ & overall amplification factor & eq.\ \eqref{eq:ampfactor}, p.\ \pageref{eq:ampfactor} & 

  \end{tabular}
  \caption{Overview of the notation and symbols used to describe the HTMM. An
    asterisk $(*)$ indicates that the corresponding parameter is usually unknown
    and needs to be estimated (jointly) from the time traces by the methods
    described in section \ref{sec:estimation}. The three parameters describing
    the influence of camera and noise can be obtained (or estimated)
    independently. The references refer to the first mention of the respective
    quantity in sections \ref{sec:fluorophore} and \ref{sec:thinning}.
  }
  \label{tab:notation}
\end{table}

\section{Estimation}
\label{sec:estimation}

In the previous sections, we have developed a statistical model for the time
series of the observable fluorescence generated by $m$ fluorophores. We now
address the central goal of this article: estimating $m$ with our
model. To this end, let
\begin{equation}
 y = (y_t)_{t=1}^T
\end{equation}
be a realization of the process $\ccdvals^{(m)}$ in \eqref{eq:full_model} that
models the observable fluorescence during the measurement process. In practice,
$y$ is obtained from a series of microscopy images (frames) by summing the
camera output values over some fixed region of interest. See
\cite{laitenberger_2018} for details on necessary or beneficial preprocessing
steps.

The fluorophore number $m$ will not be the only unknown parameter of
$\ccdvals^{(m)}$. Indeed, several (or even all) of the parameters $\gamma = (m,
q_{00}, \nu_0, \alpha^0, \alpha^1, \lambda, \theta)$ that describe the
first two moments of $\ccdvals^{(m)}$ (see Table \ref{tab:notation}) are usually
not known precisely, since the properties of the fluorophore heavily depend on
the fluorophore type itself and on details of the experimental setting.  The
preferable choice is therefore to jointly estimate all values in $\gamma$,
respecting the constraints that are inherent to the model.\footnote{
  See the end of Section \ref{sec:fluorophore} for a general discussion of
  these constraints, and Section \ref{sec:alexa} for a discussion in context
  of the fluorophore Alexa 647.}

\paragraph{Pseudo log-likelihood.}
The process $\ccdvals^{(m)}$ has a complicated non-Gaussian and
non-stationary structure with long term correlations. 
Furthermore, it is essentially impossible to evaluate the likelihood function
numerically for a given set of parameters as it consists of too many terms
(see Remark \ref{rem:complexity}).
This makes direct likelihood-based methods to estimate the model parameters
$\gamma$ unsuitable. 
To overcome this difficulty, we approximate $\ccdvals^{(m)}$ by a Gaussian
process with known parametric form of the expectation $\mu = \mu(\gamma)$ and
covariance $\Sigma = \Sigma(\gamma)$, see equations
(\ref{eq:final-moments}a-d). This leads to the \emph{pseudo log-likelihood}
\begin{equation}\label{eq:pseudologlikelihood}
  \tilde{l}_y(\gamma) = -\frac{1}{2}\big[(y - \mu) \,
  \Sigma^{-1}\, (y - \mu) + \log\det\Sigma \big],
\end{equation}
where we neglect an additive constant that would belong to the full
log-likelihood of the Gaussian process.
We estimate the model parameters $\gamma$ by finding a set of values
$\hat{\gamma}$ that maximize this expression,
\begin{equation}\label{eq:gammahat}
  \hat{\gamma} = \argmax_{\gamma\in\,\Gamma}\, \tilde{l}_y(\gamma).
\end{equation}
While this approach significantly simplifies the estimation compared to
direct treatment of $\ccdvals^{(m)}$, equation \eqref{eq:gammahat} still
represents a non-convex optimization problem over a parameter space
$\Gamma$ that obeys several (non-linear) constraints. As such, there is
neither a closed theory nor a canonical method for numerical treatment
available.

In the general case, where all parameters in $\gamma$ need to be estimated and
no additional constraints can be posed, $\Gamma$ will be a manifold of dimension
$3r+3$ (see the discussion at the end of Section \ref{sec:fluorophore}).
For specific choices of the inner model, there could be fewer free parameters in
$\theta$, reducing the dimension of $\Gamma$. 
If $\nu_0$ is known and not equal to 1 or 0, there is 1 parameter less. If
it is known and equal to 1 or 0, then there are even $r$ free parameters
less, because $\alpha^1$ respectively $\alpha^0$ drop out of the
expressions in \eqref{eq:final-moments}.
In case of Alexa 647, with a model of $r=3$ dark states and an additional
constraint on $\theta_2$, see Section \ref{sec:alexa}, we are thus confronted
with an $8$, $10$, or $11$-dimensional parameter space $\Gamma$.

\paragraph{Numerical procedure.}
Finding a numerical solution of the optimization problem \eqref{eq:gammahat}
poses several challenges. First, the high dimensionality of $\Gamma$ in
combination with both equality and inequality constraints forces one to apply
very general optimization schemes (like (quasi) Newton methods,
primal-dual-splitting, or nonlinear conjugate gradient methods). Some
schemes rely on gradient information about $\tilde{l}_y(\gamma)$
while others are gradient-free. All of them, however, work in a local fashion
and thus crucially rely on the choice of suitable initial parameters
$\gamma_\mathrm{init}$.
Indeed, optimizing over all parameters of $\gamma$ simultaneously was
empirically found to depend sensitively on the initial values and did not always
converge to the global maximum. Instead, approaches where different components
of $\gamma$ were held fixed at times -- and partial optimizations with
methods like the simplex-search algorithm by Nelder and Mead \cite{nelder1965}
were conducted sequentially -- turned out to be more successful in practice.

To find suitable initial parameters, different methods can be applied. One
option is to first employ a multi exponential fit of the expectation value
$\mu_t$. This will yield first guesses for $\lambda$ and for the product
$m\theta_1\alpha$. However, this fit may be of poor quality if the number $r$ of
dark states is large. The value of $\theta_1$ may furthermore be guessed from
late segments in the time traces, where with high probability at most one
fluorophore is active due to bleaching.
In case of the experimental data analyzed in \cite{laitenberger_2018}, we
eventually found a set of initial parameters that worked well on a range of
different image series in experimentally similar conditions.

\paragraph{Estimation results.}
In order to demonstrate that the proposed pseudo log-likelihood approach works
in principle, we apply it to estimate $\gamma$ for simulated traces. We use the
same model choices like in Figure \ref{fig:simulations} and
\ref{fig:covariances} and consider the case $\nu_0 = 1$. The true parameter
values in this setting are given by
\begin{equation*}
  \theta_1 \approx 767, \quad \theta_2 \approx 0.95, \quad \theta_3 \approx
  0.056, \quad q_{00} \approx 0.9, \\
\end{equation*}
as well as
\begin{equation*}
  \lambda \approx (0.99, 0.89, 0.86),\quad \alpha^0 \approx (0.05, 1.23, -0.28).
\end{equation*}
Including the molecule number $m$, there are 8 degrees of freedom in total. 
We contrast two different choices of initial parameters:
the true parameters $\gamma^1_\mathrm{init} = \gamma$, and an arbitrary
selection $\gamma^2_\mathrm{init}$ determined by
\begin{equation}\label{eq:badinitials}
  \theta_1 = 1, \quad \theta_3 = 0.1, \quad \lambda = (0.99, 0.9, 0.8),\quad
  \text{and}\quad m\,\theta_1\alpha^0 = (1, 1, 1),
\end{equation}
which yields parameters very different from the true ones (with the
exception of $\lambda$).

\begin{figure}
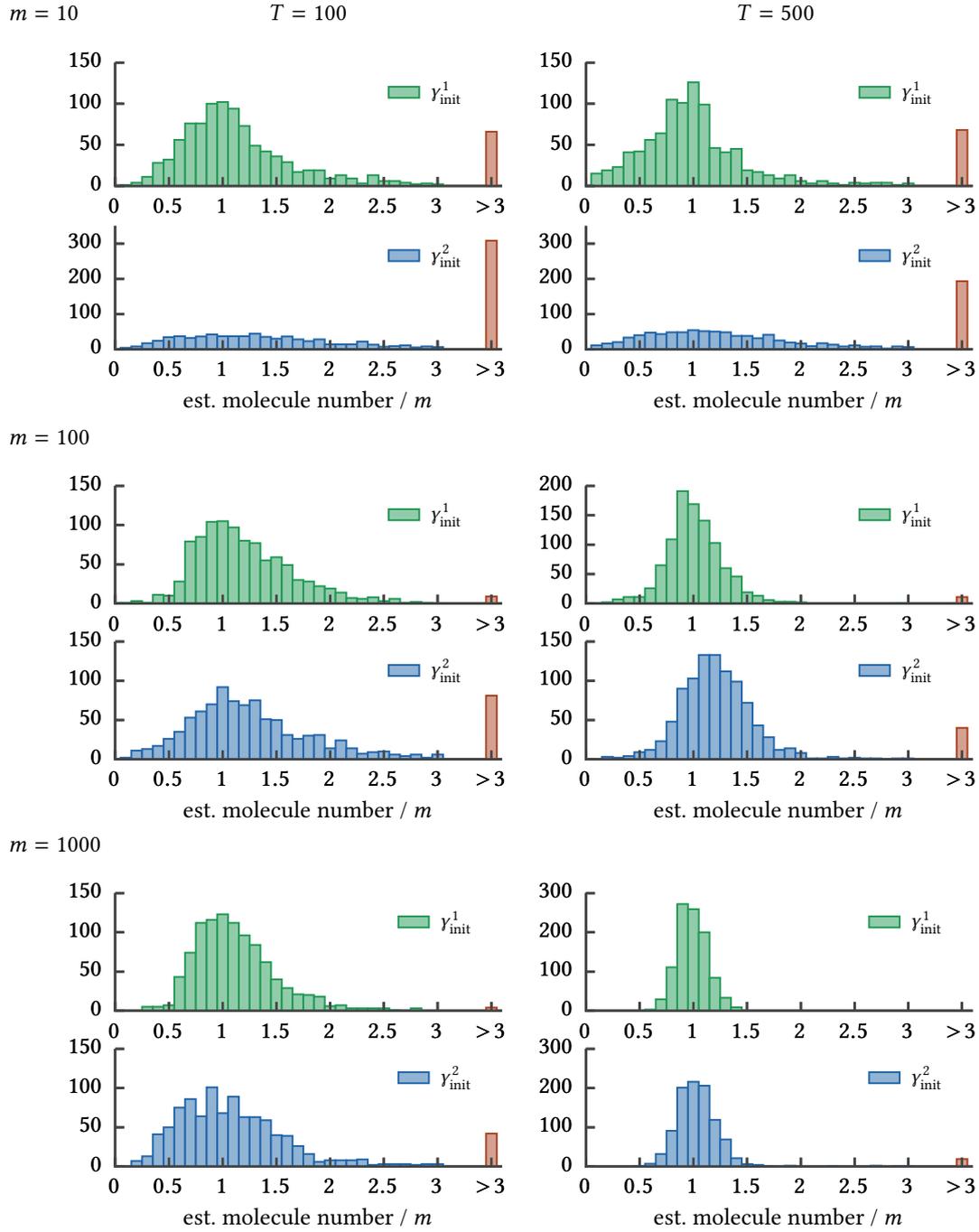

  \centering\small
  \begin{minipage}{\linewidth}$m = 10$\hspace{2.7cm} $T = 100$\hspace{5.6cm} $T = 500$\end{minipage}

  \vspace{0.3cm}
  \input{counting-10.tex}
  \vspace{-0.3cm}

  \begin{minipage}{\linewidth}$m = 100$\end{minipage}

  \vspace{0.3cm}
  \input{counting-100.tex}
  \vspace{-0.3cm}

  \begin{minipage}{\linewidth}$m = 1000$\end{minipage}

  \vspace{0.3cm}
  \input{counting-1000.tex}
  \caption{Estimated molecule number for simulated fluorescence traces of length
    $T = 100$ (left column) as well as $T = 500$ (right column). For each
    $m\in\{10, 100, 1000\}$ and initial parameters $\gamma_\mathrm{init}^{1}$
    (true values, green) or $\gamma_\mathrm{init}^{2}$ (\enquote{wrong} values
    defined by \eqref{eq:badinitials}, blue), we show histograms obtained from
    the empirical distribution of the pseudo log-likelihood estimator for 1000
    independent repetitions. Since the estimator sometimes failed to converge
    to reasonable solutions (usually resulting in exploding values for $m$),
    we collect estimates $> 3m$ in the red bar on the right of each graph.}
  \label{fig:counting}
\end{figure}

The estimation results for $m \in \{10, 100, 1000\}$ under the observation of $T
\in \{100, 500\}$ time points of the simulated process $\photons^{(m)}$ are
depicted in Figure \ref{fig:counting}. One can see that the estimates are
largely reasonable, especially when $m$ is large or when the initial parameters
are set to the oracle choice $\gamma^1_\mathrm{init}$. On the other hand, the
estimator evidently struggles for small values of $m$ and poor initial
parameters $\gamma^2_\mathrm{init}$, where it even fails to provide the right
magnitude of $m$ in about one fourth of the repetitions. Generally, providing
longer fluorescence traces ($T = 500$ instead of $100$) improves the estimator
notably -- even though it also results in a more pronounced bias under
$\gamma^2_\mathrm{init}$ for $m = 100$.

The estimation of other parameters, like the average photon number $\theta_1$,
the excess relative variance $\theta_3$, or the probability $q_{00}$ that the
fluorophore stays in the bright state during exposure, performed similarly to
the estimation of $m$.
However, we often observed that the estimator severely struggles to
guess the values of $\alpha$ and $\lambda$ correctly, particularly for small
values of $m$. This is not surprising, since the coefficients and exponents for
the multi-exponential decay $\sum_{x\in\states} \alpha_x\lambda_x^{t-1}$ are
very hard to identify even under generous conditions; let alone from few noisy
molecules. 

We want to emphasize that the estimation scheme introduced here is preliminary
and can be improved in several ways. 
First, there are various (local) optimization algorithms that typically perform
superior to the basic Nelder-Mead method we applied. For example, making use of
gradients of the pseudo log-likelihood, which can be calculated explicitly, could
greatly improve the performance and runtime of the estimation process.
Furthermore, our results indicate that prior knowledge -- in form of suitable
initial parameters or prior distributions on the parameters -- have 
beneficial effects on the estimation results. Indeed, many of the failed
estimations arise due to convergence to minima far from the true parameters.
The adoption of Bayesian methods therefore appears promising to yield more
stable and reliable estimators.

\section{Case Study: Alexa 647}
\label{sec:alexa}
In the following, we specialize our general model for the fluorophore
Alexa 647, which was used in the experimental work of
\cite{laitenberger_2018}. Fluorophores of the Alexa series are popular in
diverse areas of biomedical research. Due to their properties, like high
photostability and brightness, they are amongst the most common choices for super
resolution microscopy markers for in vitro cell experiments.
Alexa 647 dyes can be used to label a wide variety of molecules, e.g., DNA
\cite{laitenberger_2018} or proteins like IgG antibodies, streptavidin, or
transferrin \cite{berlier_2003}.
They have their absorption maximum at $650\,\mathrm{nm}$ and their emission
maximum at $671\,\mathrm{nm}$.

\paragraph{Markov model.} 
We focus on the inner and outer models that are illustrated in Figure
\ref{fig:model-alexa}. 
According to the statistical fluorophore model $\singlemodels$ established in
Section \ref{sec:fluorophore}, we need several components in order to describe
the full fluorophore behavior: the inner and outer transition matrices
$\innermodel$ and $\outermodel$, represented by the values $Q = (q_{xz})$,
and the photon statistics $p_{00}(y)$ and $p_{10}(y)$. 
The outer model depicted in Figure \ref{fig:model-alexa} has $r = 3$ dark
states, and the respective transition matrix is given by
\begin{equation*}
  \outermodel = \begin{pmatrix}
    1 & q_{01} & q_{02} & 0 \\
    0 & q_{11} & q_{12} & 0 \\
    0 & q_{21} & q_{22} & 0 \\
    0 & q_{31} & q_{32} & 1 \\
  \end{pmatrix}.
\end{equation*}
The inner model, in contrast, is given through a Markov chain with inner states
$\mathrm{S}_0$, $\mathrm{S}_1$, and $\mathrm{D}_3$. We furthermore include an
``exit'' state $\mathrm{E}$ (see Remark \ref{rem:commonexit}) in our description, which
indicates the transition to one of the dark states of the outer model. This is
the only possibility for the fluorophore to leave the bright state during
exposure.  The short-time matrix $\innermodel$ for $r = 3$ is
\begin{equation*}
  \innermodel = \begin{pmatrix}
    q_{00} & 0 & 0 & 0 \\
    q_{10} & 1 & 0 & 0 \\
    q_{20} & 0 & 1 & 0 \\
    q_{30} & 0 & 0 & 1 \\
  \end{pmatrix}.
\end{equation*}
Photons are emitted in this inner model if the fluorophore makes the transition
from the excited singlet state $\mathrm{S}_1$ to the ground state $\mathrm{S}_0$. 

Note that this model for Alexa 647 is not an exact representation of
the quantum mechanical state diagram of the fluorophore, which would require the
inclusion of a high number of states that each possess vibrational and
rotational substates. However, the full details of the quantum mechanical energy
landscape are unknown for most fluorophores, and, according to Remark
\ref{rem:compoundstates}, they also do not have to be known, since states that
live on similar timescales can be identified. In fact, it appears from
the empirical study in \cite{laitenberger_2018} that the description
provided above captures the essential features of the observed behavior of Alexa
647.
Furthermore, this Markov model -- or slight modifications thereof -- should be
appropriate to model other fluorophores, too.

\paragraph{Photon statistics.}
We now derive the photon statistics $p_{00}$ and $p_{10}$ for
a single exposure. Our derivation is not based on a rigorous
treatment of the inner Markov chain in Figure \ref{fig:model-alexa} but on
a reasonable approximation.
For convenience, we will use the symbol $\photons$ to refer to
a random variable that has distribution $\photons_t\,|\,\prestate_t = 0$ in
the following. This means that $Y$ will \emph{not} refer to the single
fluorophore process $(Y_t)_{t=1}^T$ for the duration of this section.

Let us call a maximal uninterrupted sequence of transitions between
$\mathrm{S}_0$ and $\mathrm{S}_1$ a ``burst''. A burst is ended by a transition
to $\mathrm{D}_3$. This leads to a geometric distribution: do the loop
$\mathrm{S}_0\to \mathrm{S}_1\to \mathrm{S}_0$ until failure $\mathrm{S}_1\to
\mathrm{D}_3$. Calling the probability of failure $\geop\in(0,1)$, it is clear that the
number of loops, and hence the number of photons in this burst, has a geometric
distribution with parameter $1 - \geop$. During each exposure interval, there
will be a number of $B$ bursts such that the total number of photons $Y$ is
a sum of $B$ independent geometrically distributed random variables. This leads
to a negative binomial distribution with parameters $B$ and $1 - p$,
\begin{align}\label{eq:negbin}
  Y&\sim \mathrm{NegBin}(B,1 - \geop).
\end{align}
The number of bursts $B$ is a random quantity. To determine its distribution,
first consider the case that occurs when the bright state can never be left.
Then the distribution of $B$ would be approximately Poissonian: dividing the
exposure interval into many small intervals, each much longer than a typical
burst but much smaller than the exposure time\footnote{
  This is possible due to the large difference in
  transition rates between $\mathrm{S}_0\to \mathrm{S}_1$ and $\mathrm{D}_3\to
  \mathrm{S}_0$.}, 
there is a small probability for a burst in each interval and a large number of
intervals such that one is in the Poisson limit of the binomial distribution.
When taking into account the transition to the exit state $\mathrm{E}$, there is
a nonzero
probability for a ``failure'' (i.e., exiting) before each burst. Hence, every
burst can be viewed as one successful trial, and the bursts continue until
either the exit state or the end of the exposure time is reached. Therefore,
the number of bursts is given by the minimum of a Poissonian and a geometric
random variable that are independent of each other,
\begin{align}
  B &=    \min(Z, Q), \nonumber\\
  Z &\sim \mathrm{Poisson}(\poi), \label{eq:microrandomvariables}\\
  Q &\sim \mathrm{Geom}(1 - \geob), \nonumber
\end{align}
with parameters $\poi > 0$ and $0 < \geob < 1$. The transition to the
exit state happens if and only if a failure happened before the end of the
exposure time. This is the case when $Z > Q$. The following result provides
a connection between the parameters $\geob$ and $\poi$ of the photon
distribution and the parameter $q_{00}$ of the transition matrix $\innermodel$. 
\begin{lemma}{}{microparams}
  The Alexa 647 fluorophore stays in the bright state during exposure with probability
  \begin{equation}\label{eq:q00_short_time}
    q_{00} = \PP(Z\le Q) = e^{-(1-\geob)\poi} .
  \end{equation}
\end{lemma}
\begin{proof}
  Direct calculation yields
\begin{align*}
\PP(Z>Q) 
  &= \sum_{z=0}^\infty \PP(Z=z) \PP(Q<z)
  = \sum_{z=0}^\infty e^{-\poi}\frac{\poi^z}{z!}\sum_{k=0}^{z-1} 
    \geob^k (1-\geob)\nonumber\\
  &= \sum_{z=0}^\infty e^{-\poi}\frac{\poi^z}{z!} (1-\geob^z) 
     = 1- \sum_{z=0}^\infty e^{-\poi}\frac{(\geob\poi)^z}{z!}\nonumber\\
  &= 1-e^{-(1-\geob)\poi} = 1 - q_{00},
\end{align*}
which shows the claim.
\end{proof}

We next derive the moment generating function $G$ of $Y$. It can be expressed
via the moment generating function $G^B$ of the number $B$ of bursts.

\begin{lemma}{}{moment_Y}
  The moment generating function of $Y$ is given by
  \begin{align}
    G(\xi) &= G^B\!\left(\log\!\left(\frac{1-\geop}{1-\geop e^\xi}\right)\right),
    \label{eq:moment_generating_function_Y}
    \shortintertext{where}
    G^B(\xi) 
    &= 
      e^{-(1-q)\poi} e^{-(1-e^\xi)q\poi}  +
      \big(1-e^{(1-q)\poi}\big) \frac{1-q}{1-qe^\xi}
      \frac{1-e^{-(1-qe^\xi)\poi}}{1 - e^{(1-q)\poi}} \nonumber \\[0.15cm]
    &= q_{00}\, G_{00}^B(\xi) + (1 - q_{00})\, G_{10}^B(\xi).
    \label{eq:moment_generating_function_B}
  \end{align}
  The functions $G^B_{00}$ and $G^B_{10}$ denote the generating functions of $B$
  conditioned on $Z \le Q$ and $Z > Q$, respectively.
\end{lemma}

\begin{proof}
Based on the distributions of the auxiliary random variables $B$, $Z$, and $Q$,
we calculate
\begin{align*}
  G^B(\xi) 
  &= \EE\big[e^{\xi B}\big] 
  = \EE\Big[\EE\big[e^{\xi B}|\, Z\big]\Big] \\
  &= \EE\left[  \sum_{k=0}^{Z-1} e^{\xi k} \geob^k(1-\geob) 
     + \sum_{k=Z}^\infty e^{\xi Z} \geob^k(1-\geob)  \right] \\
  &= \EE\left[ (1-\geob)\frac{1-(\geob e^\xi)^Z}{1-\geob e^\xi} 
     + (qe^\xi)^Z \right] \\
  &= \big(1-e^{(1-q)\poi}\big) \frac{1-q}{1-qe^\xi}
  \frac{1-e^{-(1-qe^\xi)\poi}}{1 - e^{(1-q)\poi}} 
     + e^{-(1-q)\poi} e^{-(1-e^\xi)q\poi}.
\end{align*}
The decomposition of $G^B$ that is implied in
\eqref{eq:moment_generating_function_B} states that $B$ given the event $Z \le
Q$ is Poisson distributed with the reduced parameter $\poi\geob$. 
Using Lemma \ref{lem:microparams}, this can be established as follows:
\begin{align*}
  \PP(Z = z \,|\, Z \le Q) &= \frac{\PP(Z = z) \,\PP(z \le Q)}{\PP(Z\le Q)} 
    = \frac{\frac{\poi^z}{z!} e^{-\poi} \sum_{k=z}^\infty \geob^k (1-\geob)}{e^{-(1-\geob)\poi}}
    = \frac{(\poi\geob)^z}{z!} e^{\poi\geob}.
\end{align*}
Finally, we have 
\begin{equation*}
  G(\xi) = \EE\big[e^{\xi \photons}\big] 
    = \EE\Big[\EE\big[e^{\xi \photons}|\,B\big]\Big]
    = \EE\left[\left(\frac{1-\geop}{1-\geop e^\xi}\right)^{\!\!B}\right]
    = G^B\!\left(\log\!\left(\frac{1-\geop}{1-\geop e^\xi}\right)\right)
\end{equation*}
for the moment generating function of $\photons$.
\end{proof}

Similar results hold for the conditional generating functions $G_{00}$ and
$G_{10}$, which are expressible by $G^B_{00}$ and $G^B_{10}$ in a likewise
fashion. Therefore, the photon statistics $p_{00}$ and $p_{10}$ are
completely specified in terms of the inner parameters $\geop$, $\geob$, and
$\poi$ by equations \eqref{eq:moment_generating_function_Y} and
\eqref{eq:moment_generating_function_B}. 

\paragraph{Moments.} In order to understand the effects of our inner
model for the second-order description of the fluorophore, we look at the
parameters $\theta_1$, $\theta_2$, and $\theta_3$ as introduced in
\eqref{eq:theta}. 
\begin{lemma}{}{alexamomentparameters}
  In the Alexa 647 inner model, the moment parameters are given by:
  \begin{subequations}
  \begin{align}
    \theta_1 &= \frac{p}{1-p}\frac{q}{1-q}(1 - q_{00}), \label{eq:theta1alexa}\\
    \theta_2 &= -\frac{q_{00}\log{q_{00}}}{1 - q_{00}}, \label{eq:theta2alexa}\\
    \theta_3 &= \frac{2}{1-q_{00}}\left(\frac{1-\geob}{\geob} - \theta_2 + 1 \right) - 1. \label{eq:theta3alexa}
  \end{align}
  \end{subequations}
\end{lemma}
\begin{proof}
  All parameters can be expressed via derivatives of the generating functions
  $G$ and $G_{00}$ at zero. We find
\begin{equation*}
  \theta_1 = \EE[Y] = G'(0) = \frac{p}{1-p}\frac{q}{1-q}(1 - q_{00})
\end{equation*}
for the unconditioned expectation. Similarly, one can calculate
\begin{equation*}
  \theta_2 
    = \frac{q_{00}\,\EE[Y\,|\,Z\le Q]}{\theta_1} 
    = \frac{q_{00}}{\theta_1} \,G'_{00}(0) 
    = -\frac{q_{00}\log{q_{00}}}{1 - q_{00}}.
\end{equation*}
Finally, the parameter $\theta_3$ can be computed to read
\begin{equation*}
  \theta_3 = \frac{\Var \,Y}{\theta_1^2} - \frac{1}{\theta_1}
    = \frac{G''(0)}{\theta_1^2} - \frac{1}{\theta_1} - 1
    = \frac{2}{1-q_{00}}\left(\frac{1-\geob}{\geob} - \theta_2 + 1 \right) - 1,
\end{equation*}
which captures the variance of $\photons$ relative to $\theta_1$.
\end{proof}

Surprisingly, according to equations \eqref{eq:theta2alexa} and
\eqref{eq:q00_short_time}, the parameter $\theta_2$ only depends on the internal
parameters via $q_{00} = e^{-(1-\geob)\poi}$.
For reference, we note that this relation can be inverted by the formula
\begin{equation*}
  q_{00} = -\frac{\theta_2}{W_{-1}(-\theta_2 \,e^{-\theta_2})},
\end{equation*}
where $W_{-1}$ is the branch of the Lambert-$W$ function whose range
contains the interval $(-\infty,-1/e)$ of the real line (see, e.g.,
\cite{corless_lambertw_1996} for a definition and a review of some
properties of this function). 
This consideration shows that even though the inner model in Figure 
\ref{fig:model-alexa} depends on three independent parameters ($\geop$,
$\geob$, and $\poi$), only two free parameters, namely $\theta_1$ and
$\theta_3$, remain in the second-order description. This must be taken
into account when formulating and conducting the optimization routine for
the pseudo log-likelihood \eqref{eq:pseudologlikelihood} used to estimate
the model parameters.

\paragraph{Invariance under thinning.} 
There is another remarkable feature of this inner model choice
that deserves to be highlighted. In Section \ref{sec:thinning}, we derived
how thinning -- the independent loss of photons with a certain
probability -- affects the parameters $\theta = (\theta_1, \theta_2,
\theta_3)$, and we concluded that $\theta_1$ is transformed to $\theta_1'
= p_\mathrm{d}\, \theta_1$ while $\theta_2$ and $\theta_3$ are left
untouched. Interestingly, we can make a much stronger statement for our
model of Alexa 647. 

\begin{lemma}{}{alexainvariance}
  The parametric family of the distribution $Y$ for Alexa 647 is left invariant
  by thinning. More precisely, the thinned process $Y' \sim
    \mathrm{Bin}(Y, p_\mathrm{d})$ obeys equations \eqref{eq:negbin} and
    \eqref{eq:microrandomvariables} defining $Y$ if the parameter $p$ is
    replaced by 
  \begin{equation*}
    p' = \frac{p\,p_\mathrm{d}}{1 - p + p\,p_\mathrm{d}},
  \end{equation*}
  while $q$ and $\mu$ remain the same.
\end{lemma}

\begin{proof}
Let $D_1, D_2, \ldots \sim\mathrm{Ber}(p_\mathrm{d})$ be
independently Bernoulli distributed, and let $G^D(\xi)$ be the moment
generating function of $D_1$, 
\begin{equation}\label{eq:moment_generating_D}
  G^D(\xi) = \EE\big(e^{\xi D_1}) = 1 + p_\mathrm{d}\,(e^\xi - 1).
\end{equation}
The thinned photon count is given through $\thinned \sim \mathrm{Bin}(Y,
p_\mathrm{d})$, which can also be stated as
\begin{equation*}
  \thinned = \sum_{i = 1}^Y D_i.
\end{equation*}
The moment generating function of $\thinned$ is
\begin{equation*}
  G^{\thinned}\!(\xi) 
    = \EE\Big[\EE\big[e^{\xi Y'}\,|\,Y\big]\Big]
    = \EE\big[e^{\log(G^D(\xi))\photons}\big] = G\big(\log G^D(\xi)\big),
\end{equation*}
where $G$ denotes the generating function of $Y$. When we define the
transformed probability 
\begin{equation}\label{eq:probtransform}
  p' = \frac{p\,p_\mathrm{d}}{1 - p + p\,p_\mathrm{d}},
\end{equation} 
we can invoke expression \eqref{eq:moment_generating_function_Y} for $G$ and
equation \eqref{eq:moment_generating_D} for $G^D$ to find
\begin{equation*}
  G^{Y'}\!(\xi) 
    = G^B\!\left(\log\!\left(\frac{1-p}{1-p
      \big(1+p_\mathrm{d}(e^\xi - 1)\big)}\right)\right)
    = G^B\!\left(\log\!\left(\frac{1-p'}{1-p'e^\xi}\right)\right).
\end{equation*}
Comparison to equation \eqref{eq:moment_generating_function_Y} reveals that the
distribution of detected photons has the same parametric form as without
thinning, and the only effect is the monotone transformation
\eqref{eq:probtransform} of the geometric probability $p$. The same statement
holds for the generating functions conditioned on $Z > Q$ or $Z \le Q$, so the
parametric family of the photon statistics for our model of Alexa 647 is indeed
left invariant by thinning.
\end{proof}

\section{Outlook}
\label{sec:outlook}

Optical nanoscopy has evolved into a scientific junction point that drives
cutting-edge research in disciplines as diverse as optics, biochemistry, and
statistics. The work we presented in this article contributes to this
development by providing a new way to statistically model the temporal activity
pattern of fluorophores, which form the basis for fluorescence super-resolution
microscopy.  More than that, however, our work is meant to help expose the wide
spectrum of worthwhile statistical, computational, and mathematical questions
that are raised by current developments in this area.

One point of immediate interest is a better understanding of the photon emission
process and its approximation by a Gaussian process. Besides heuristic hints --
like the practical success for the purpose of estimating the molecule number --
we have not yet found analytical guarantees of how well the Gaussian
approximation captures the essential properties of the original HTMM
investigated in this article, or if other methods of estimation could offer
improvements.
Furthermore, we lack statistical results for the estimation process via maximum
likelihood estimation of the pseudo log-likelihood. In this light, the
derivation of (asymptotic) statements and confidence bounds for Gaussian
processes that are constrained like in our case pose an interesting challenge.
This kind of insights would also help address another prevalent question: to
what extent is it possible to infer the HTMM model parameters in situations
where not all fluorophores are bright in the beginning of the experiment,
meaning $\nu_0 < 1$? Conclusive results in this regard are still pending.
Reliable inference in this scenario, however, could be an important step towards
more flexibility in the design of quantitative super-resolution experiments.
A different aspect that is emphasized by our observations in Section
\ref{sec:estimation}, but that we only touched upon in passing, is the
beneficial effect of prior knowledge for the estimation performance. This
underscores the crucial role that Bayesian estimation methods routinely play for
the statistical modeling of biological system. For a simplified, binomial model
to count fluorophores, the advantage of Bayesian approaches was recently
demonstrated in \cite{schneider2018}.

Another set of problems that our work calls attention to, especially
\cite{laitenberger_2018}, is the modeling of dependencies in biomolecular
systems. Even though the description of complex molecules via Markov chains has
been proven to be highly successful, the appropriate statistical description of
interactions between them, and the resulting influence on measured data in
experiments, is largely unresolved. Future research may focus on this
question and investigate dependency mechanisms for hidden Markov models, like
the proposed HTMM, which are simple enough for analytical investigation but are
still able to capture dependency structures suggested by experimental
observations.

\begin{appendix}

\section{Lumped Markov Chains}
\label{sec:lumped}

Like for most biochemical compounds, the precise quantum physical state space of
fluorophores is likely to be much more involved than our relatively simple four
(outer) state model for Alexa 647 (see Section \ref{sec:alexa}).
Still, modeling fluorophores and other biochemical molecules as Markov chains
with a smaller number of states than actual quantum states works often very
well. In the following, we briefly outline that states with the same transition
rates can be combined into single states without losing the Markov property. Our
exposition follows \cite{kemeny_1960}.

Let $\states$ denote the finite state space of a stationary Markov chain
$X = (X_t)_{t\in\mathbb{N}}$ with transition Matrix $M$. Let 
$P = \{\mathcal{S}_1,\ldots, \mathcal{S}_k\}$ be a partition of $\states$, and
let $\pi\colon\states
\to P$ be the respective projection map, meaning $\pi(x) = \mathcal{S}_i$
iff $x\in\mathcal{S}_i$ for some $i = 1, \ldots, k$. The stochastic process
defined by $\pi\circ X$ is called the \emph{lumped process}, where only
states can be observed which are merged according to the partition $P$.
Furthermore, the chain $X$ is called \emph{lumpable} with respect to $P$ if 
$\pi\circ X$ is again a Markov chain for any initial distribution on $X$.
Similarly, $X$ is called \emph{weakly lumpable} if the lumped process is
Markov for some initial distribution on $X$.

\begin{theorem*}{Kemeny and Snell, 1960}
  A necessary and sufficient condition for a Markov chain to be lumpable with
  respect to a partition $P$ is that for every pair of sets $\states_i$ and
  $\states_j$ in $P$, the probabilities
  \begin{equation*}
    p_{\states_jx} = \sum_{z\in\states_j} M_{zx}
  \end{equation*}
  to transition to $\states_j$ from $x$ have the same value for every
  $x\in\states_i$. These common values form the transition matrix for the
  lumped chain $\pi\circ X$.
\end{theorem*}
\begin{proof}
  The proof can be found in \cite{kemeny_1960}, Theorem 6.3.2. Note that Kemeny
  and Snell consider row-stochastic transition matrices while we use
  column-stochastic ones.
\end{proof}

In the case of interest to us, the partition $P$ of $\mathcal{S}$ 
corresponds to groupings of physical fluorophore states that have (approximately)
the same transition rates to states of other groups. It is easy to check that
the conditions for the above theorem are satisfied in this setting and that the
lumped chain is thus Markov.

\begin{corollary*}{}
  Let $P$ be a partitioning $\{\mathcal{S}_1, \ldots, \mathcal{S}_k\}$ of states
  $\states$ such that for all $i= 1, \ldots, k$
  \begin{equation}\label{eq:sametransitions}
    M_{zx} = M_{zx'} \qquad \text{for all}\qquad
    \text{$x,x'\in\states_i$ and $z\in\states\setminus \states_i$}. 
  \end{equation}
  Then the Markov chain $X$ is lumpable with respect to $P$.
\end{corollary*}
\begin{proof}
For $i \neq j$, it immediately follows from \eqref{eq:sametransitions} that
\begin{equation*}
  p_{\states_jx} = \sum_{z\in\states_j} M_{zx} = 
  \sum_{z\in\states_j} M_{zx'} = p_{\states_jx'}
\end{equation*}
for all $x, x'$ in $\mathcal{S}_i$. If $i = j$, we similarly find
\begin{equation*}  
  p_{\states_i x}
    = 1 - \sum_{z\in\states\setminus\states_i} M_{zx}
    = 1 - \sum_{z\in\states\setminus\states_i} M_{zx'}
    = p_{\states_i x'}.
\end{equation*}
Therefore, $p_{\mathcal{S}_jx}$ does not
depend on $x\in\mathcal{S}_i$ for all $j = 1,\ldots,k$, and the theorem of
Kemeny and Snell can be applied.
\end{proof}

\section{Diagonalizability of Stochastic Matrices}
\label{sec:diag}

Let $\mathcal{M}$ be the set of stochastic $n\times n$ matrices with a fixed $n
\ge 2$. A (column) stochastic matrix $M \in \mathcal{M}$ has only non-negative entries
and its columns sum to 1. Therefore, the set $\mathcal{M}$ forms
an $n(n-1)$-dimensional submanifold of the unit cube $[0,1]^{n\times n}$. In the
following, we want to show that it is reasonable to assume that a randomly
picked matrix in $\mathcal{M}$ is diagonalizable with probability one. 

To this end, we consider the following projection map $\pi$ from $[0,1]^{n\times
n}$ to $\mathcal{M}$. We represent values in $[0,1]^{n\times n}$ by $v = (v_1,
\ldots, v_n)$, where each $v_i$ is a column vector in $[0,1]^n$. Let
$\tilde{v}_i$ denote the normed vector
\begin{equation*}
  \tilde{v}_i = \frac{v_i}{\sum_{j=1}^n v_{ij}}.
\end{equation*}
Then, for almost all $v\in [0, 1]^{n\times n}$ with respect to the Lebesgue
measure $\lambda$ on $[0, 1]^{n\times n}$, we can define
\begin{equation*}
  \pi(v) = (\tilde{v}_1, \ldots, \tilde{v}_n) \in \mathcal{M}.
\end{equation*}
Note that $\lambda$ is a probability measure on $[0, 1]^{n\times n}$.
Since $\pi$ is measurable (assuming the subspace Borel $\sigma$-field on
$\mathcal{M}$), we obtain the induced probability measure $\lambda^\pi$ on
$\mathcal{M}$, where 
\begin{equation*}
  \lambda^\pi (V) = \lambda\big(\pi^{-1}(V)\big)
\end{equation*}
for any measurable $V \subset \mathcal{M}$.

\begin{lemma}{}{diag}
  Let $N \subset \mathcal{M}$ be the set of stochastic matrices that are not
  diagonalizable. Then 
  \begin{equation*}
    \lambda^\pi (N) = 0,
  \end{equation*}
  i.e., any matrix in $\mathcal{M}$ is diagonalizable with probability one.
\end{lemma}

\begin{proof}
  In the following, we only consider matrices in the open domain $(0,1)^{n\times
  n}$, which have strictly positive entries. This is sufficient, since the
  boundary $[0,1]^{n\times n} \setminus (0,1)^{n\times n}$ has 
  Lebesgue measure
  zero. We also define the stochastic matrices $\mathcal{M}^+ = \pi\big((0,
  1)^{n\times n}\big)$ with strictly positive entries.
  
  Let $\Delta: \RR^{n\times n} \to \RR$ be the function that maps $A \in \RR^{n\times
  n}$ to the discriminant of the characteristic polynomial of $A$. The
  characteristic polynomial of $A$ is given by 
  \begin{equation*}
    \lambda \mapsto \mathrm{det}\big(A - \lambda\mathbbm{1}_n\big),
  \end{equation*}
  where $\mathbbm{1}_n$ is the identity matrix in $n$ dimensions, and the
  discriminant of a polynomial with roots $\lambda_1, \lambda_2, ..., \lambda_n
  \in\CC$ is proportional to
  \begin{equation*}
    \prod_{i<j} \big(\lambda_i - \lambda_j\big)^2.
  \end{equation*}
  Since the coefficients of the characteristic polynomial are polynomials in the
  entries of $A$, and since the discriminant is a real polynomial of the
  coefficients (see, e.g., \cite{basu_2007}), $\Delta$ itself is a real
  polynomial in the entries of $A$. It is easy to check that $\Delta$ is not
  constant $0$ on $\mathcal{M}^+$.

  The map $\pi$ is real-analytical on $(0, 1)^{n\times n}$, and the map $\Delta$ is
  real-analytical as a function restricted to (the real-analytical manifold)
  $\mathcal{M}^+$. Consequently, the composition 
  \begin{equation*}
    \Delta\circ\pi: (0, 1)^{n\times n} \to \RR
  \end{equation*}
  is a non-constant real-analytic function. Due to the properties of the
  discriminant, we have that
  \begin{equation*}
    \Delta(M) \neq 0 \quad \Longleftrightarrow 
      \quad \text{$M$ has $n$ distinct eigenvalues},
  \end{equation*}
  where the latter property implies diagonalizability. Thus, if $N^+
  \subset \mathcal{M}^+$ denotes the set of non-diagonalizable positive
  stochastic matrices, we find
  \begin{equation*}
    \lambda^\pi(N^+) \le \lambda^\pi\big( \Delta^{-1}(0) \cap \mathcal{M}^+ \big) 
      = \lambda\big( (\Delta\circ\pi)^{-1}(0) \big) = 0.
  \end{equation*}
  The last equality follows from the fact that the set of roots of
  non-zero real-analytic functions on any connected open domain of $\RR^d$,
  $d\in\mathbb{N}$, has $\lambda^d$-Lebesgue measure zero \cite{mityagin_2015}.
  Since the difference between $N$ and $N^+$ is only a null set, the statement
  of the lemma follows.
\end{proof}

\section{Real and Positive Eigenvalues}
\label{sec:eig}

In Appendix \ref{sec:diag}, we showed that essentially every transition matrix $M$
in our fluorophore model $\singlemodels$ is diagonalizable.
Here, we argue that it is even plausible to assume real positive eigenvalues, i.e.,
that the spectrum $\sigma(M)$ of $M$ is contained in $[0,1]$
if the number $r$ of dark states is lower than or equal to $3$. The
corresponding matrices for $r = 1, 2, 3$ look like
\begin{equation}\label{eq:M123}
  M_1 = 
  \begin{pmatrix}
    a_1 & 0 \\
    a_2 & 1 
  \end{pmatrix},
  \qquad
  M_2 = 
  \begin{pmatrix}
    a_1 & a_2 & 0 \\
    a_3 & a_4 & 0 \\
    a_5 & a_6 & 1
  \end{pmatrix},
  \qquad
  M_3 = \begin{pmatrix}
    a_1    & a_2    & a_3    & 0 \\
    a_4    & a_5    & a_6    & 0 \\
    a_7    & a_8    & a_9    & 0 \\
    a_{10} & a_{11} & a_{12} & 1
  \end{pmatrix},
\end{equation}
where $a_i \in [0,1]$ for all $i = 1, ..., r(r+1)$, and where all columns sum up
to one.

The case $r = 1$ is trivial: the eigenvalues are $a_1$ and $1$.  For $r = 2$ and
$3$, we have to make further assumptions in order to conclude
$\sigma(M)\subset [0,1]$. If $r = 2$, we need that the diagonal
values of $M$ are large enough (see Lemma \ref{lem:eig2} below), and for $r = 3$
we additionally require that the diagonal values are sufficiently distinct, as is
made precise in Lemma \ref{lem:eig3}.  These two assumptions -- large and
distinct diagonal values -- are natural for our setting: the former means that
the outer states are usually stable over more than one frame, while the latter
holds if states with similar dwell times are merged for the Markovian description
of the fluorophore (see Remark \ref{rem:compoundstates} and Appendix
\ref{sec:lumped} for more details).

\begin{lemma}{}{eig2}
  The eigenvalues of the matrix $M_2$ given by \eqref{eq:M123} are real with
  absolute value $\le 1$. If all diagonal entries are $\ge 1/2$, they are
  additionally non-negative, such that
  \begin{equation*}
    \sigma(M_2) \subset [0,1].
  \end{equation*}
\end{lemma}
\begin{proof}
  Since $M_2$ is a stochastic matrix, the absolute value of all eigenvalues of
  $M_2$ is bounded by 1. The characteristic polynomial $Q$ of $M_2$ is
  \begin{equation*}
    Q(\lambda) = (1 - \lambda) \,[(a_1 - \lambda)(a_4 - \lambda) - a_2a_3] = (1-\lambda)
    \,[\lambda^2 - (a_1 + a_4)\,\lambda + (a_1a_4 - a_2a_3)].
  \end{equation*}
  Clearly, $\lambda = 1$ is an eigenvalue.
  In order to ensure that all other eigenvalues are also real, we investigate 
  the discriminant $\Delta$ of the second factor above (in square brackets). It
  is given by
  \begin{equation*}
    \Delta = (a_1+a_4)^2 - 4\,(a_1a_4 - a_2a_3) = (a_1 - a_4)^2 + 4 \,a_2a_3,
  \end{equation*}
  which satisfies $\Delta > 0$, since all entries are non-negative.
  Consequently, all eigenvalues are real-valued. If additionally $a_1, a_4 \ge
  1/2$ (and thus $a_2,a_3\le 1/2$), we find
  \begin{equation*}
    Q(\lambda) > a_1a_4 - a_2a_3 \ge 0
  \end{equation*}
  for $\lambda < 0$. Therefore, all eigenvalues must be non-negative.
\end{proof}

The most important application of our theory is the Alexa 647 model described in
Section \ref{sec:alexa}. Here, $r = 3$, which leads to a transition matrix of the
form $M_3$. Unlike for $r = 1, 2$, these matrices do not always have real
eigenvalues. 
The following lemma provides an analytical criterion for all eigenvalues to be real.

\begin{lemma}{}{eig3criterion}
  Assume that $M_3$ as given in \eqref{eq:M123} is diagonalizable and that its upper
  left $3\times 3$ submatrix is irreducible.
  Then, besides the value $1$, $M_3$ has a second real eigenvalue
  $\lambda_0\in(0, 1]$ that is larger than or equal to the diagonal entries,
  \begin{equation*}
    \max\{a_1, a_5, a_9\} \le \lambda_0.
  \end{equation*}
  The remaining two eigenvalues of $M$ are real if and only if
  \begin{equation}\label{eq:realcondition}
    \big(\bar{a}_1 + \bar{a}_5 + 
         \bar{a}_9\big)^2 + 4\,\big(a_6a_8 + a_2 a_4 + a_3 a_7 - 
         \bar{a}_1\bar{a}_5 - \bar{a}_1\bar{a}_9 -
         \bar{a}_5\bar{a}_9\big) \ge 0,
  \end{equation}
  where $\bar{a}_i = a_i - \lambda_0$.
\end{lemma}

\begin{proof}
  Let $M'$ be the upper left $3\times 3$ submatrix of $M_3$. The characteristic
  polynomial of $M_3$ is given by
  \begin{equation}\label{eq:charpoly3}
      \mathrm{det}(M_3 - \lambda \mathbbm{1}_{4}) 
      = -\,(1-\lambda) \cdot \mathrm{det}(M' - \lambda \mathbbm{1}_{3}),
  \end{equation}
  where $\mathbbm{1}_{3}$ and $\mathbbm{1}_{4}$ are the identity matrices in
  three and four dimensions, respectively. We can therefore restrict our study
  to the eigenvalues of $M'$, which is an irreducible matrix with non-negative
  entries by assumption. This allows us to apply Perron-Frobenius theory
  \cite{meyer_2000}. In particular, $M'$ has a real eigenvalue $0 < \lambda_0 \le 1$,
  such that all other eigenvalues of $M'$ are smaller in absolute value. 
  Since $M'$ is diagonalizable, this largest eigenvalue $\lambda_0$ is equal to the
  operator norm of $M'$. In particular, $\lambda_0$ is thus larger than all diagonal
  entries of $M'$,
  \begin{equation}\label{eq:l0greater}
    \max\{a_1, a_5, a_9\} \le \lambda_0,
  \end{equation}
  because $a_1 = e_1^\mathrm{T}M'e^{}_1 \le \lambda_0$, where $e_1 = (1, 0,
  0)^\mathrm{T}$, and analog relations hold for $a_5$ and $a_9$.
  We denote the characteristic polynomial of $M'$ by $Q$, and we also define the
  shifted polynomial
  \begin{equation*}
    H(\mu) = Q(\mu + \lambda_0).
  \end{equation*}
  Then $H(0) = Q(\lambda_0) = 0$, since $\lambda_0$ is an eigenvalue of $M'$. Setting
  $\bar{a}_i = a_i - \lambda_0$, Sarrus' rule for the determinant of
  $3\times 3$ matrices yields
  \begin{align*}
    H(\mu) 
      &= \mathrm{det}(M' - (\mu + \lambda_0)\,\mathbbm{1}) \nonumber\\
      &= (\bar{a}_1 - \mu) (\bar{a}_5 - \mu) (\bar{a}_9 - \mu) + a_2 a_6 a_7 + a_3 a_4 a_8
        - a_6 a_8 (\bar{a}_1-\mu) - a_3 a_7 (\bar{a}_5 - \mu) - a_2 a_4 (\bar{a}_9 - \mu) \nonumber\\
      &= -\mu^3 + (\bar{a}_1 + \bar{a}_5 + \bar{a}_9) \mu^2 + (a_6a_8 + a_2 a_4 + a_3 a_7 - \bar{a}_1\bar{a}_5 - \bar{a}_1\bar{a}_9 - \bar{a}_5\bar{a}_9)\mu,
  \end{align*}
  where we harnessed in the last step that the constant part of the polynomial
  $H$ vanishes due to $H(0) = 0$. The polynomial $H$ (and thus $Q$) has
  exclusively real roots if and only if the discriminant of $H$ is
  non-negative. The discriminant of a polynomial $ax^3 + bx^2 + cx + d$ of order
  three is given by
  \begin{equation*}
    \Delta = b^2c^2 - 4ac^3 -4b^3d - 27a^2d^2 + 18abcd.
  \end{equation*}
  In our case, where $d = 0$ and $a = -1$, we find
  \begin{equation*}
    \Delta/c^2 = b^2 - 4ac = \big(\bar{a}_1 + \bar{a}_5 + \bar{a}_9\big)^2 + 4\,\big(a_6a_8 + a_2 a_4 + a_3 a_7 - \bar{a}_1\bar{a}_5 - \bar{a}_1\bar{a}_9 - \bar{a}_5\bar{a}_9\big)
  \end{equation*}
  for the relevant part of the discriminant of $H$, since $c^2 > 0$ can be
  assumed (if $c = 0$, $\Delta = 0$ follows and all eigenvalues are real).
\end{proof}

It is easy to find examples, where condition \eqref{eq:realcondition} is
violated, and eigenvalues are complex. For instance, the matrix
\begin{equation*}
  M_3 =
  \begin{pmatrix}
    0.8 & 0   & 0.1 & 0 \\
    0.1 & 0.8 & 0   & 0 \\
    0.1 & 0.1 & 0.8 & 0 \\
    0   & 0.1 & 0.1 & 1 
  \end{pmatrix}
\end{equation*}
has the eigenvalues $\lambda \approx 1, 0.93, 0.73\pm0.06\,i$. The upper left
$3\times 3$ block in this matrix is irreducible, but the value in condition
\eqref{eq:realcondition} is about $-0.013$.
Therefore, we need another restriction for $M_3$ in order to establish real
eigenvalues.

Fortunately, a gap condition for the diagonal of $M_3$ does the trick.
Assume that the diagonal values $a_1, a_5, a_9$ of $M_3$ are distinct. We order
these values by magnitude and denote them by $d_1 > d_2 > d_3$.
Let 
\begin{equation}\label{eq:mudef}
  \mu_1 = \frac{\lambda_0 - d_1}{\lambda_0 - d_2},
    \qquad
  \mu_2 = \frac{\lambda_0 - d_2}{\lambda_0 - d_3},
\end{equation}
where $\lambda_0\in(0, 1]$ as in Lemma \ref{lem:eig3criterion}. Due to
\eqref{eq:l0greater}, we find that $\mu_1, \mu_2 \in [0, 1)$. The 
condition we need for a real spectrum is
\begin{equation}\label{eq:diversity}
  \mu_2^2 (1 - \mu_1)^2 \ge 2\mu_2(1 + \mu_1) - 1.
\end{equation}

\begin{lemma}{}{eig3}
  If condition \eqref{eq:diversity} is satisfied and Lemma
  \ref{lem:eig3criterion} can be applied, the eigenvalues of the matrix $M_3$
  given by \eqref{eq:M123} are all real, with absolute values bounded by one.
  If all diagonal entries of $M_3$ are $\ge 2/3$, it additionally holds that the
  spectrum is non-negative,
  \begin{equation*}
    \sigma(M_3) \subset [0, 1].
  \end{equation*}
\end{lemma}

\begin{remark}{gap condition}{gap}
  Inequality \eqref{eq:diversity} really poses a gap condition on the
  diagonal values, and therefore on the lifetimes of the outer states of the
  fluorophore. If, for example, $\mu_2 \le 1/4$, or if both $\mu_1$ and $\mu_2$
  are smaller than $3/8$, condition \eqref{eq:diversity} will hold. Values of
  $\mu_1$ and $\mu_2$ close to $1$, on the other hand, violate the inequality.
  In this context, further note that condition \eqref{eq:diversity} is
  sufficient for Lemma \ref{lem:eig3} even if we replace $\lambda_0$ by $1$ in
  definition \eqref{eq:mudef} of $\mu_1$ and $\mu_2$. We can therefore conclude
  that $\sigma(M_3) \subset [0, 1]$ is guaranteed as long as the probabilities
  $1 - d_i$ to leave an outer state in a time step are sufficiently diverse,
  i.e., the respective dwell times must be on different timescales. For
  instance, diagonal values
  \begin{equation*}
    d_1 = 0.975, \qquad d_2 = 0.95, \qquad d_3 = 0.8,
  \end{equation*}
  which correspond to lifetimes of about 40, 20, and 5 microscopy frames on
  average, yield $\mu_1 = 1/2$ and $\mu_2 = 1/4$ for $\lambda_0 = 1$. 
  Thus, condition \eqref{eq:diversity} is satisfied and any transition
  matrix $M_3$ with these diagonal values meets the gap requirement for
  Lemma \ref{lem:eig3}.
\end{remark}

\begin{proof}[Proof of Lemma \ref{lem:eig3}]
  In order to show that all eigenvalues are real, we consult Lemma
  \ref{lem:eig3criterion}. Condition \eqref{eq:realcondition} is certainly
  satisfied if
  \begin{equation*}
    \big(\bar{a}_1 + \bar{a}_5 + \bar{a}_9\big)^2 - 
    4\,\big(\bar{a}_1\bar{a}_5 + \bar{a}_1\bar{a}_9 + \bar{a}_5 \bar{a}_9\big) \ge 0.
  \end{equation*}
  Using the definition of $d_i$ as ordered diagonal values of $M_3$ for
  $i=1,\ldots,3$, this is equivalent to
  \begin{equation*}
    (1 + \mu_2 + \mu_1\mu_2)^2 - 4\,(\mu_2 + \mu_1\mu_2 + \mu_1\mu_2^2) \ge 0,
  \end{equation*}
  where $\mu_1$ and $\mu_2$ are defined as in \eqref{eq:mudef}. This inequality
  is equivalent to \eqref{eq:diversity}, which can be shown by straightforward
  computation. 
  
  It remains to be shown that all eigenvalues are non-negative if each diagonal
  entry of $M$ is $\ge 2/3$. For this, assume
  that there would be an eigenvalue $\lambda < 0$ of $M_3$. Due to the form
  \eqref{eq:charpoly3} of the characteristic polynomial of $M_3$, $\lambda$ must
  also be an eigenvalue of $M'$, the upper left $3\times 3$ submatrix of $M_3$.
  Let $v\in\RR^3$ denote a normalized eigenvector of $M'$ to $\lambda$. Then
  \begin{equation}\label{eq:neglambda}
    v^T M' v = \lambda < 0
  \end{equation}
  must hold. On the other hand, if we denote the non-bleached states by $\states
  = \{0,1,2\}$, we have
  \begin{equation*}
    \left(\sum_{x\in\states} |v_x| \right)^2 \le 3 \sum_{x\in\states} v_x^2 = 3
  \end{equation*}
  by Jensen's inequality, and consequently
  \begin{equation}\label{eq:sumineq}
    \sum_{x,z\in\states, x\neq z} |v_x||v_z|
      = \left(\sum_{x\in\states} |v_x| \right)^2 - \sum_{x\in\states} v_x^2
      \le 2.
  \end{equation}
  We can thus establish
  \begin{align*}
    v^T M' v 
      &= \sum_{x,z\in\states} v_x M'_{xz} v_z \\
      &= \sum_{x\in\states}   v_x M'_{xx} v_x 
         + \sum_{x,z\in\states, x\neq z} v_x M'_{xz} v_z \\
      &\ge \frac{2}{3} \sum_{x\in\states} v_x^2 
         - \frac{1}{3} \sum_{x,z\in\states, x\neq z} |v_x| |v_z| \\
      &= 0,
  \end{align*}
  where we used that $M'_{xx} \ge 2/3$ and hence $M'_{xz} \le 1/3$ for $x\neq z$,
  and applied inequality \eqref{eq:sumineq}.
  This contradicts \eqref{eq:neglambda}, which is why all eigenvalues must be
  positive.
\end{proof}

\end{appendix}

\section*{Acknowledgments}
This work was in part supported by the German Science Foundation (DFG) through
grant CRC 755 “Nanoscale Photonic Imaging”, projects A4 and A6, RTG 2088
“Discovering Structure in Complex Data”, and Germany's Excellence Strategy EXC
2067/1-390729940. We are furthermore grateful to Mira Jürgens for computational
assistance and proofreading.

\printbibliography

\end{document}